\documentclass[11pt,reqno]{amsart}
\usepackage{geometry}
\usepackage{amsfonts}
\usepackage{amsmath}
\usepackage{amsthm}
\usepackage{amssymb}
\usepackage{mathrsfs}

\usepackage{mathtools} 
\usepackage{enumitem}  
\usepackage{lmodern}
\usepackage{microtype}
\usepackage{setspace}
\usepackage{subcaption}
\usepackage{graphicx}
\usepackage{booktabs}

\usepackage[ruled,vlined,linesnumbered]{algorithm2e}
\usepackage{algorithmic}
\usepackage{xcolor}

\usepackage{geometry}
\usepackage{hyperref}
\hypersetup{
	colorlinks = true,
	allcolors  = blue,  
}

\mathtoolsset{showonlyrefs} 
\numberwithin{equation}{section}
\allowdisplaybreaks

\newenvironment{cdate}{
	\noindent
	\textbf{\small This version:}
	\small
}{\par}

\newenvironment{keys}{
	\noindent
	\textbf{\small Keywords:}
	\small
}{\par}

\newenvironment{msc}{
	\noindent
	\textbf{\small MSC 2020 codes:}
	\small
}{\par}

\newcommand{\hemail}[1]{\href{mailto:{#1}}{\normalfont \textnormal{#1}}}

\theoremstyle{plain} 
\newtheorem{theorem}{Theorem}[section]

\newtheorem{lemma}[theorem]{Lemma}

\newtheorem*{notation}{Notation}

\theoremstyle{definition} 
\newtheorem{discussion}[theorem]{Discussion}
\newtheorem{definition}[theorem]{Definition}
\newtheorem{example}[theorem]{Example}

\newtheorem{SA}[theorem]{Standing Assumption}

\newtheorem{remark}[theorem]{Remark}
\newtheorem*{FM}{Financial Market Model}

\numberwithin{equation}{section}

\DeclareMathOperator{\FD}{FD} 

\newcommand{\rd}{\mathrm{d}}
\newcommand{\vd}{\,\mathrm{d}}

\newcommand{\on}{\operatorname}
\newcommand{\m}{m}
\newcommand{\s}{s}
\newcommand{\bR}{\mathbb{R}}

\newcommand{\1}{\mathbf{1}} 
\newcommand{\cF}{\mathcal{F}}
\newcommand{\Y}{S}

\renewcommand{\S}{S}
\newcommand{\X}{X}

\renewcommand{\l}{\alpha}
\renewcommand{\r}{\beta}

\newcommand{\I}{J}

\newcommand{\tm}{\widetilde{\m}}
\newcommand{\hm}{\nu} 

\newcommand{\G}{g}
\newcommand{\g}{G}
\newcommand{\om}{\overline{\m}}
\newcommand{\oS}{\overline{\Y}}
\newcommand{\cH}{H} 
\newcommand{\tY}{\widetilde{\Y}}
\renewcommand{\varkappa}{\kappa}
\newcommand{\bW}{\mathbb{W}}
\newcommand{\cW}{\mathcal{W}}

\newcommand{\rr}{\overline{u}}
\newcommand{\J}{J^*}

\DeclareMathOperator{\IR}{\mathbb{R}}

\newcommand{\ol}[0]{\overline}
\newcommand{\mc}[0]{\mathcal}
\newcommand{\wh}[0]{\widehat}
\newcommand{\wt}[0]{\widetilde}

\newcommand{\indic}[1]{\mathbf{1}_{#1}}
\newcommand{\indicB}[1]{\mathbf{1}_{\{#1\}}}
 
\newcommand{\sqbraces}[1]{ \left[{#1}\right] }

\title[Hedging in general diffusion markets]{On the hedging problem \\ in general 1D diffusion markets}

\author[A. Anagnostakis]{Alexis Anagnostakis}
\address{
	A. Anagnostakis -- Centro de Modelamiento Matemático (CNRS IRL2807), Universidad de Chile, Santiago, Chile.
}
\email{\hemail{alexis.anagnostakis@yandex.com}, \hemail{aanagnostakis@cmm.uchile.cl}}

\author[D. Criens]{David Criens}
\address{D. Criens -- University of Freiburg, Ernst-Zermelo-Str. 1, 79104 Freiburg, Germany.}
\email{\hemail{david.criens@stochastik.uni-freiburg.de}}

\author[M. Urusov]{Mikhail Urusov}
\address{M. Urusov -- University of Duisburg-Essen, Thea-Leymann-Str. 9, 45127 Essen, Germany.}
\email{\hemail{mikhail.urusov@uni-due.de}}

\begin{document}

	\begin{abstract}
		We develop a PDE-based methodology for pricing and hedging European contingent claims in general one-dimensional diffusion markets characterized solely by their scale function and speed measure, possibly without a classical SDE representation, and with constant interest rate. 
		We derive a hedging equation whose solution generates a self-financing hedging strategy and provide
		sufficient conditions on scale, speed, and interest rate, under which this strategy achieves the minimal hedging capital, expressed through the no free lunch with vanishing risk (NFLVR) condition.
		We further prove necessary and sufficient conditions for NFLVR and characterize the class of equivalent local martingale measures through an auxiliary diffusion whose scale and speed characteristics are determined by those of the real-world diffusion and by the interest rate. When the NFLVR condition fails, the framework may produce multiple hedging equations corresponding to non-minimal strategies, whose associated prices can exceed the minimal hedging capital. We illustrate both the effectiveness and limitations of the approach through numerical experiments involving diffusion models with irregular features.
	\end{abstract}

	\maketitle
	
	\begin{quote}
		\begin{cdate}
			\today
		\end{cdate}
		
		\begin{msc}
			91G20, 60J60, 91G80.
		\end{msc}
		
		\begin{keys}
			Pricing; hedging contingent claims; generalized second order differential operators; general diffusion; scale and speed; non-zero interest rates.
		\end{keys}
	\end{quote}

	\section{Introduction}
	
	Pricing and hedging of contingent claims are among the most important problems in mathematical finance.
	One fundamental method of solving these problems is the {\em partial differential equation (PDE) methodology} that essentially traces back to groundbreaking contributions of Black--Scholes and Merton, awarded with the 1997 Nobel prize in economics. 
	Traditionally, the method is carried out for market models with risky assets  whose price processes \(\Y = (\Y_t)_{t \geq 0}\) solve a stochastic differential equation (SDE) of the form
	\begin{align} \label{eq: SDE} 
		\rd \Y_t = b (\Y_t) \vd t + \sigma (\Y_t) \vd W_t, \quad \text{where \(W\) is a Brownian motion},
	\end{align} 
	with sufficiently smooth coefficients \(b\) and \(\sigma\), and a bank account process \(\Y^0_t = e^{rt}\) with a deterministic interest rate \(r \in \bR\). 
	For a given payoff function \(h \in C_b (\bR)\) and a finite time horizon \(T > 0\), let \(u\) be the classical solution to the fundamental pricing PDE 
	\begin{align} \label{eq: fundamental PDE SDE model}
		\begin{cases}
			u_t (t, x) 
			+ \frac{\sigma^2 (x)}{2} u_{xx} (t, x)
			+ rx\, u_x (t, x) 
			- r u(t, x) = 0, & (t, x) \in [0, T) \times \bR,
			\\[2mm]
			u (T, x) = h (x),& \ x \in \bR.
		\end{cases} 
	\end{align}
	In this SDE framework it is well-known that, under suitable regularity and growth assumptions, \(u (0, \Y_0)\) is the minimal hedging capital of the European contingent claim with payoff \(h (\Y_T)\), and that the payoff can be realized through the hedging strategy
	\begin{equation}
		H^{(0)}_t = \frac{u (t, \Y_t) - u_x (t, \Y_t) \Y_t}{\S^0_t}, \quad H^{(1)}_t = u_x (t, \Y_t).
	\end{equation}
	We refer to \cite{HS_00, R_13} for derivations under suitable regularity assumptions on the coefficients, and \cite[Section~2.2.3]{CJY}, \cite[Section~10.3]{pas_11}, or \cite[Section~VII.4.c]{shir} for textbook treatments.
	At this point, we also like to mention the papers
		\cite{BayraktarXing2010,
			Cetin2018,
			CetinLarsen2023,
			EkstromTysk2009}
		that investigate the possible non-uniqueness issues of the fundamental pricing PDE for unbounded payoff functions.
	
	In this paper, we investigate the PDE methodology in the context of market models that remain Markovian but no longer admit an SDE representation of the form \eqref{eq: SDE}. 
	Unlike SDE solutions, price dynamics in such markets can exhibit a wide array of singular features like partially reflecting thresholds (skewness, see \cite{Lejay2006}), spending a positive amount of time at some point(s) (stickiness, see \cite{EngPes, salins2017markovprocesses}), and fractal slowdowns (see \cite[Section~8]{ankirchner2020_a_functional_limit}). 
	More specifically, we consider a financial market model with a single risky asset whose price \(\Y\) follows a
	\emph{regular continuous strong Markov process}
	(alternatively called
	\emph{diffusion in the sense of It\^o--McKean} \cite{itokean74}),
	and with the same bank account dynamics as in the above SDE model, that is, \(\Y^0_t = e^{rt}\) for a constant interest rate $r\in \IR $. 
	Our model is parametric in the sense that the market is uniquely characterized by three objects: the characteristics that determine the law of the diffusion $S$ (its scale function \(\s\) and speed measure~\(\m\)), and the interest rate $r$. 
	For concreteness, we focus on diffusions with open state space \(J\), such as \(J = (0, \infty), J = \bR\), or \(J = (0, 1)\), and we assume that the scale function \(\s\) and its inverse \(\s^{-1}\) are dc functions (i.e., the difference of two convex functions). 
	These assumptions imply that \(\Y\) is a semimartingale (see \cite[Section~5]{CinJPrSha}), which aligns with the usual paradigms of financial modeling. 
	
		Our main contributions are as follows.
		First we derive what we call the \emph{hedging PDE}
		(which is, in general, a non-standard PDE)
		--- a certain evolution equation related to an \emph{auxiliary diffusion} whose scale and speed depend on the market characteristics $(s,m,r)$.
		We emphasize that even its precise form is a novel result in settings without classical SDE dynamics of the type~\eqref{eq: SDE}.
		We separately discuss both the existence and the uniqueness of what we call ``good'' solutions to this equation.
		For the uniqueness we require an additional assumption,
		which is essential in the sense that we can construct examples with multiple good solutions when that assumption is not satisfied.
		The hedging PDE always yields a hedging strategy, but in general it may not necessarily achieve the minimal hedging capital. A sufficient condition for achieving the minimal hedging capital is the existence of an \emph{equivalent local martingale measure (ELMM)} for the discounted risky asset, which is equivalent to the classical \emph{no free lunch with vanishing risk (NFLVR)} no-arbitrage condition by the fundamental theorem of asset pricing of Delbaen and Schachermayer \cite{DS}.
		We provide a characterization of NFLVR in terms of the model parameters $(s,m,r)$
		and show that under any ELMM the {\em un}discounted price process follows the law of the auxiliary diffusion that we mentioned above in the context of the hedging PDE. 

		Our theoretical findings are complemented by numerical illustrations. We test our methodology on diffusion models featuring two skew-sticky interfaces, covering various parameter configurations: outward-pointing skew, inward-pointing skew, and skew points with the same direction. We consider a bear spread option
        and choose the interest rate $r$ so that the first two models admit an ELMM; the third model, however, fails to satisfy the ELMM condition for any interest rate due to its skew structure. We implement our hedging approach on these three models, together with the Bachelier model as a benchmark. Despite numerical instabilities arising from our approximation methods --- finite difference solutions of the hedging PDE and their delta fields --- the hedging methodology proves to be consistent: the tracking error decreases as the number of portfolio rebalancing instances increases and also as the precision of the finite-difference approximations increases. This holds irrespective of whether the market admits an ELMM.
	
	\medskip
	
	We end this introduction with comments on related literature. In the context of the hedging problem for diffusion market models with irregular coefficients, we are only aware of the recent paper \cite{anagnostakis2025pricing} that studies a sticky Black--Scholes model without interest rate (this model satisfies NFLVR only if the constant short rate is null). The paper establishes the pricing PDE for its model and provides numerical illustrations.
	
	Our work appears to be the first systematic approach for the pricing and hedging problem in the general setting with singular features, even in the case of zero interest rate $r=0 $.
	The NFLVR condition for general diffusion market models with $r=0 $ has been investigated in the recent papers \cite{CU_FS_25, CU_AAP_25}. 
	To our knowledge, we provide the first results in the presence of interest rates.
	In this connection, it is worth noting that our setting with $r\ne0$ does not reduce to it with $r=0$ because, in general, the discounted price $S/S^0$ loses the time homogeneity.
		We, finally, highlight that the structure of the ELMM as a diffusion in the sense of It\^o--McKean with particular scale and speed characteristics appears to be a nontrivial and novel observation.

	\subsection{Paper outline}
	
	Section~\ref{sec: hedging_PDE} presents our market model along with the main results of this paper, which form the theoretical core of our analysis.
	Section~\ref{sec: examples} presents examples of general diffusion markets that exhibit various features of interest: skew-sticky thresholds
	and accessible boundaries for the hedging PDE (exit, regular).
	Section~\ref{sec: numexp} presents numerical experiments illustrating our main results on the examples from Section~\ref{sec: examples}, demonstrating the practical feasibility and the good statistical properties of our approach.
	Section~\ref{sec: proofs} contains the proofs of the results presented in Section~\ref{sec: hedging_PDE}.

	\subsection{Notation}
	
	Throughout the paper we adopt the following notation.
	For a function $f \colon \IR \to \IR $, its right- and left-hand derivatives are denoted by
	\begin{equation}
		f'_+(x) := \lim_{h\searrow 0} \frac{f(x+h)-f(x)}{h} \quad \text{and} \quad 
		f'_-(x) := \lim_{h\searrow 0} \frac{f(x)-f(x-h)}{h},
	\end{equation}
	respectively, whenever the limits exist.
	For the derivative, we simply write $f'(x)$
	whenever both $f'_+(x)$ and $f'_-(x)$ exist and are equal.
	We recall that, for a
	\emph{dc function}
	(difference of two convex functions) \(f\colon \mathbb{R} \to \mathbb{R}\), its right-hand derivative \(f'_+\) exists and is RCLL (right continuous with left-hand limits),
	its left-hand derivative \(f'_-\) exists and is LCRL (left continuous with right-hand limits),
	and it holds
	\[
	f'_-(x) = (f'_+)(x-)
	\quad\text{and}\quad
	f'_+(x) = (f'_-)(x+).
	\]
	Further, let \(u\colon [0,T] \times J \to \mathbb{R}\) be a real-valued function of time and space.
	\begin{enumerate}
		\item
		The time derivative is denoted by
		\[
		\frac{\partial u}{\partial t} (t,x) 
		:=
		\begin{cases}
			\displaystyle
			\lim_{h \to 0} \frac{u(t+h,x) - u(t,x)}{h}, & (t, x) \in (0, T) \times J,
			\\[3mm]
			\displaystyle
			\lim_{h \searrow 0} \frac{u(h,x) - u(0,x)}{h}, &
			(t,x) \in \{0\}\times J, 
		\end{cases} 
		\]
		whenever the limit exists.
		We do not need the time derivative on
		$\{T\}\times J$.
		
		Consistently with~\eqref{eq: fundamental PDE SDE model}, we sometimes write $u_t $ in lieu of $\partial u / \partial t$.
		
		\item
		For a strictly increasing continuous function
		\(F\colon J \to \mathbb{R}\), the right- and left-hand Stieltjes derivatives with respect to \(F\) in the space argument are defined as
		\[
		\frac{\partial^+ u}{\partial F} (t,x) := \lim_{h\searrow 0} \frac{u(t,x+h) - u(t,x)}{F(x+h)-F(x)}
		\]
		and
		\[
		\frac{\partial^- u}{\partial F} (t,x) := \lim_{h\searrow 0} \frac{u(t,x) - u(t,x-h)}{F(x)-F(x-h)},
		\]
		respectively, whenever the limits exist.
		As for the functions of one variable, we write
		$\partial u / \partial F$
		and speak about the two-sided Stieltjes derivative
		in the space argument
		at some point $(t,x)$
		whenever $\partial^+ u/\partial F$
		and $\partial^- u/\partial F$
		coincide at $(t,x)$.
		
		Finally, in the case $F=\on{id}$, we write
		$\partial^+_x u$ in lieu of
		$\partial^+ u / \partial x$,
		and
		$\partial^-_x u$ in lieu of
		$\partial^{-} u / \partial x$,
		and,
		consistently with~\eqref{eq: fundamental PDE SDE model},
		$u_x$ in lieu of
		the two-sided derivative
		$\partial u / \partial x$.
	\end{enumerate}

	\section{Setting, main results, and examples}
	\label{sec: hedging_PDE}
	
	\subsection{Market model}
	\label{subsec: market}
	Take an open interval \(J = (\l, \r)\) with \(- \infty \leq \l < \r \leq \infty\), and let \(\Y = (\Y_t)_{t \geq 0}\) be a regular continuous strong Markov process (\emph{general diffusion} for short)
with state space $J$
	on a filtered probability space \((\Omega, \cF, (\cF_t)_{t \geq 0}, P)\) with a right-continuous filtration; see \cite[Definition~V.45.1]{RW2} for the definition on the canonical path space.
We emphasize that the strong Markov property is understood with respect to the filtration $(\cF_t)_{t \geq 0}$.
We further assume that \(S_0 = s_0\) for an arbitrary \(s_0 \in J\).
	A quite complete overview on the theory of general diffusions can be found in the seminal monograph \cite{itokean74} by It\^o and McKean.
	Shorter textbook introductions are given in \cite{freedman,kallenberg,RY,RW2}.
	
	It is well-known (\cite{itokean74}) that the law of \(\Y\) is uniquely characterized by two deterministic objects, the {\em scale function} \(\s \colon J \to \bR\) and the {\em speed measure} \(\m\).
	The scale function is a continuous increasing function on $J $, the speed measure is a measure on \((J, \mathcal{B}(J))\) such that \(\m ([a, b]) \in (0, \infty)\) for all \(a, b \in J, a < b\), and they are defined up to an affine transformation in the sense that $(b + s/a, a\, m) $ and $(s,m) $, with $a >0 $ and $b\in \IR $, define the same law. 
	Scale and speed are called the {\em diffusion characteristics}. 
	By Feller's test for explosion (\cite[Proposition~16.43]{breiman1968probability}), and since the diffusion is conservative (no killing is allowed), the assumption that \(J\) is open is equivalent to 
	\begin{equation} \label{eq: feller test} \begin{split}
			b \in \{\l, \r\}, \, |\s (b)| < \infty \ \implies \ \int_I \,  | \s &(x) - \s (b) | \, \m (\rd x) = \infty, \\& \forall \, I = (c, d) \subset J \text{ with \(b \in \{c, d\}\)}.
		\end{split} 
	\end{equation} 
	Throughout this paper, we also impose the following assumption on the scale function:
	
	\begin{SA} \label{SA}
		The scale function \(\s\) and its inverse \(\s^{-1}\) are both dc functions (i.e., the difference of two convex functions).
	\end{SA}
	
	\begin{FM}
		We are now working with a financial market that contains one risky asset~\(\Y = (\Y_t)_{t \geq 0}\) and a bank account process \(\Y^0_t := e^{rt}\) with a deterministic interest rate \(r \in \bR\). 
	\end{FM} 
	
	\begin{remark} \label{rem. SA}
		(a) Standing Assumption~\ref{SA} implies that the diffusion \(\Y\) is a semimartingale; see \cite[Section~5]{CinJPrSha}. This is considered to be a minimal assumption for financial modeling. 
		
		\smallskip
		(b) To put Standing Assumption~\ref{SA} into the context of finance, it can be considered as a mild structural assumption. Indeed, in the zero interest rate regime (i.e., when \(r = 0\)), it is even implied by the very weak ``no strong arbitrage'' (sometimes also called \(\textit{NA}_+\)) no arbitrage condition; see \cite[Corollary~4.3]{ACU_25_arxiv}. 
		
		It is worth noting that it may happen that a financial market model satisfies the even weaker no arbitrage condition ``no increasing profit'' while \(\s\) is not a dc function; see \cite[Examples~6.3, 6.4]{ACU_25_arxiv}. 
		
		\smallskip
		(c) The assumption that \(\s^{-1}\) is a dc function entails that the one-sided derivatives \(\s'_+\) and \(\s'_-\) are strictly positive. Indeed, as dc functions are locally Lipschitz continuous (as convex functions have this property), for every compact set \(K \subset J\), we obtain that 
			\[
			| x - y | \leq L | \s (x) - \s (y) |,
			\]
			for all \(x, y \in K\), where \(L > 0\) is the Lipschitz constant of \(\s^{-1}\) on the compact set \(\s (K)\).
			This proves that \(\s'_\pm \geq 1 / L\) on the interior of \(K\).
	\end{remark} 
	
	\subsection{Main results}
	Before we state our main results in detail, let us provide a short overview, meant to be a guideline for the reader. Our first main result, Theorem~\ref{theo: main1}, provides what we term the fundamental hedging PDE for European contingent claims in the market model introduced above. This is the backward evolution equation 
	given by
	\begin{equation} \label{eq: hedging PDE}
		\sqbraces{ \frac\partial{\partial t} + \frac{1}{2} \frac{\partial}{\partial\om} \frac{\partial^-}{\partial\g} - r } u(t,x) = 0, \quad 
		u(T,x) = h(x),
	\end{equation}
	where $h$ is the claim's payoff, and $(\g,\om)$ are characteristics of an auxiliary general diffusion constructed from the market parameters $(s,m,r)$.
	The theorem states, in particular,
	that a stochastic Feynman--Kac-type representation based on the auxiliary diffusion with characteristics $(\g,\om)$
	provides a so-called good solution to the hedging PDE
	and yields a hedging strategy for European contingent claims, while the solution itself gives the respective hedging capital of that strategy.
	In Theorem~\ref{theo: uniqueness pricing eq} we ensure
	the uniqueness for~\eqref{eq: hedging PDE}
	under a certain additional condition
	(the latter cannot be dropped).
	
	Our next main result relates this hedging strategy to the minimal hedging capital, i.e., the lowest superhedging price. 
	More precisely, in Theorem~\ref{theo: ELMM} we first characterize the NFLVR no-arbitrage notion through necessary and sufficient conditions on $ (s,m,r)$. Adapting martingale arguments, under NFLVR, Theorem~\ref{thm: identification} shows that the solution to~\eqref{eq: hedging PDE} yields the minimal hedging capital, which is then the no-arbitrage price of the claim. In this case, we also call~\eqref{eq: hedging PDE} the fundamental pricing PDE. 
	In Section~\ref{sec: examples} we discuss explicit examples in which case \eqref{eq: hedging PDE} either provides the minimal hedging capital or a hedging capital that is not minimal
	(for the latter, necessarily, NFLVR fails).
	In the latter case we call~\eqref{eq: hedging PDE} the non-minimal hedging PDE.	
	The proofs of our main results are given in Section~\ref{sec: proofs} below.
	
	We now start our main program, introducing the auxiliary characteristics \((\g, \om)\). The first building block of our construction is the measure
	\[
	\hm (\rd x) := - r x \s'_+ (x) \, \m (\rd x) \quad \text{on } \mathcal{B} (J), 
	\] 
	which is a locally finite signed measure. 
	From now on, we work under the following:
	
	\begin{SA} \label{SA: nu cond}
		\(\hm (\{x\}) > - \frac{1}{2}\) for all \(x \in J\).
	\end{SA}
	
	\begin{remark}\label{rem:220826a1}
		Although Standing Assumption~\ref{SA: nu cond} has a technical flavor at this point, it can be related to no-arbitrage theory: Thanks to \cite[Theorem~3.1]{ACU_25_arxiv}, it is implied by the very weak ``no increasing profit'' (NIP) condition. To provide the details, using item (ii) from  \cite[Theorem~3.1]{ACU_25_arxiv}, for every \(x \in J\), NIP implies that 
		\begin{align} \label{eq: (ii) from weak arbitrage paper}
			r x \, \m ( \{ x \} ) = \tfrac{1}{2} (\s^{-1})'' ( \{ \s (x) \} ), 
		\end{align} 
		where \((\s^{-1})''\) denotes the second derivative measure of the inverse scale function \(\s^{-1}\). By Standing Assumption~\ref{SA} (more precisely, see Remark~\ref{rem. SA}~(c)), we obtain that 
		\begin{align} \label{eq: second derivative estimate}
			(\s^{-1})'' ( \{ \s (x) \} ) = \frac{1}{\s'_+ ( x )} - \frac{1}{\s'_- ( x)} < \frac{1}{\s'_+ ( x )}. 
		\end{align} 
		Putting \eqref{eq: (ii) from weak arbitrage paper} and \eqref{eq: second derivative estimate} together yields that 
		\[
		\nu ( \{ x \} ) = - r x \s'_+ (x) \m (\{ x \}) = - \frac{1}{2} \s'_+ (x) (\s^{-1})'' ( \{ \s (x) \} ) > - \frac{1}{2}, 
		\] 
		which means that Standing Assumption~\ref{SA: nu cond} holds.
	\end{remark}
	Let \(\xi \in J\) be an arbitrary point in the state space.
	Under the above standing assumption, the integral equation
	\[
	\G (x) = \begin{cases}1 + \int_{[\xi, x]} 2 \G(z-) \, \hm (\rd z), & x \geq \xi, \\ 1 - \int_{(x, \xi)} 2 \G (z-) \, \hm (\rd z), & x < \xi, \end{cases} 
	\] 
	has a unique strictly positive RCLL solution. Indeed, the well-known formula for the stochastic exponential (or see \cite[p.~159]{ES1991}) yields that 
	\begin{align*}
		\G (x) = \begin{cases} \exp \{ 2 \hm ([\xi, x]) \} \prod_{\xi \leq y \leq x} (1 + 2\hm (\{y\})) \exp \{- 2\hm (\{y\}) \}, & x \geq \xi, \\\exp \{-2\hm ((x, \xi) ) \} \prod_{x < y < \xi} (1 + 2\hm (\{y\}))^{-1} \exp \{ 2 \hm (\{y\}) \}, & x < \xi. \end{cases} 
	\end{align*} 
	Notice that \(\G\) is locally of finite variation, which follows from the integral representation.
	
	With these objects at hand, we now define the fundamental objects that determine the hedging PDE:
	\begin{itemize}
		\item the function
		\[
		\g (x) := \int_\xi^x \G (y) \vd y, \quad x \in J, 
		\] 
		which is a strictly increasing dc function. Clearly, we also have 
		\[
		\g'_+ (x) = \G (x) \quad \text{ and }\quad \g'_- (x) = \G (x-),
		\]
		and the second-derivative measure of \(\g\) is given by \(2\G(x-) \, \nu (\rd x)\); and
		
		\item the measures
		\begin{equation}
			\label{eq_mutilde_mubar}
			\overline{\m} (\rd x) := \frac{\s'_+ (x)}{\G (x)} \, \m (\rd x) \ \ \text{on } \mathcal{B} (J), 
			\quad \tm := \overline{\m} \circ \g^{-1} \ \ \text{on } \mathcal{B} (G (J)).
		\end{equation} 
	\end{itemize}
	
	We form the interval \(\J \supset \I\) from \(\I\) by adding accessible boundary points for the characteristics
		\((\g , \om )\),
		where the accessibility is determined by Feller's test for explosion (see \cite[Proposition~16.43]{breiman1968probability}). Furthermore, we extend \(\om\) to \(\J\), choosing the
		(instantaneous or slow) reflection or absorption in regular boundary points in an arbitrary manner. Using the extension of \(\om\), we also extend \(\tm\) in the obvious way to \(\mathcal{B} (G (J^*))\).

	\begin{notation}[Canonical Setup]
		Let \(\bW\) be the space of continuous functions \(\bR_+ \to \bR\) and \(\X = (\X_t)_{t \geq 0}\) be the coordinate process on \(\bW\), i.e., \(X_t (\omega) = \omega (t)\) for all \(\omega \in \bW, t \in \bR_+\). Define \(\cW := \sigma (\X_t, t \geq 0)\) and \(\cW_t := \sigma (\X_s, s \leq t)\). The filtered space \((\bW, \cW, (\cW_{t+})_{t \geq 0})\) is called the canonical setup.  
	\end{notation}

	By \cite[Theorem~33.9]{kallenberg}, there exist two families \((J^* \ni x \mapsto Q_x)\) and \((\g(J^*) \ni x \mapsto P_x)\) of probability measures on the canonical setup that are general diffusions as defined in \cite[Definition~V.45.1]{RW2} with characteristics \((G, \om)\) and \((\on{id}, \tm)\), 
	respectively. 
	It is useful to note that the diffusions are related by the formula 
	\begin{equation}\label{eq:220826a3}
		Q_x = P_{\g (x)} \circ  \big(\g^{-1} (\X) \big)^{-1}, \quad x \in J. 
	\end{equation}
	As we will see in Theorem~\ref{theo: ELMM} below, the diffusion \((x \mapsto Q_x)\) is a natural candidate for a canonical equivalent local martingale measure (ELMM) for our market model, provided an ELMM exists. We call it \emph{auxiliary diffusion} in the sequel.
    In what follows, $T>0$ denotes a given finite time horizon.
	
	\begin{definition} \label{def: good solution}
		We call a function \(u \colon [0, T] \times J \to \bR\) a {\em good solution} to the PDE 
		\begin{align} \label{eq: backward PDE}
			\frac{\partial u}{\partial t} + \frac{1}{2} \frac{\partial}{\partial \om} \frac{\partial^- u }{\partial \g} - ru = 0 \quad \text{ on }
			[0, T) \times J, 		
		\end{align} 
		if it has the following properties: 
		\begin{enumerate}
			\item[(a)] \(u\) is continuous on \([0, T) \times J\); 
			\item[(b)] the time-derivative \(u_t\)
			exists as a continuous function on \([0, T) \times J\); 
			\item[(c)] the left-hand space derivative
			\(\partial^-_x u\)
			exists as a locally bounded function on \([0, T) \times J\) that is continuous in the first variable and LCRL in the second variable, when the other variable remains fixed; 
			\item[(d)] for all \(t \in [0, T), y, z \in J, y > z\), 
			\begin{align} \label{eq: main PDE}
				\frac{\partial^- u}{\partial\g} (t, y)
				-
				\frac{\partial^- u}{\partial\g} (t, z)
				&=
				\int_{[z, y)}
				2\, \big(
				r u (t, v) - u_t (t, v)
				\big)
				\,
				\om (\rd v).
			\end{align} 
			Notice that 
			\[
			\frac{\partial^- u}{\partial \g}(t,x) = \frac{\partial^-_x u (t,x)}{G'_-(x)} = \frac{\partial^-_x u (t,x)}{g (x -)}.  
			\]
		\end{enumerate}
	\end{definition}

	We define 
	\begin{align*}
		L^\infty (\overline{\m}) := \Big\{
		f \colon \J \to \bR
		&
		\text{ Borel-measurable such that }
		\\
		&
		\exists
		\text{ a constant }
		c<\infty
		\text{ with }
		\om \, (|f| > c) = 0
		\Big\},
	\end{align*}
	and fix a function \(h \in L^\infty (\overline{\m})\). Consider a contingent claim \(C\) with payoff \(h (\Y_T)\) at time $T$. We define the {\em fundamental value function} of this contingent claim by
	\begin{equation}\label{eq:240826a1}
	u (t, x) := E^{Q_x} \big[ e^{- r (T - t)} h (\X_{T - t}) \big], \quad (t, x) \in [0, T] \times J. 
	\end{equation}
	The following theorem explains the connection of \(u\) to the pricing and hedging problem of the contingent claim~\(C\).
	
	\begin{theorem} \label{theo: main1}
		\begin{enumerate} 
			\item[\textup{(a)}] The fundamental value function \(u\) is a good solution to \eqref{eq: backward PDE} with \(u (T, x) = h (x)\) for \(x \in J\).
			\item[\textup{(b)}] Let \(v\) be a good solution to \eqref{eq: backward PDE}. Then, almost surely, for all \(t \in [0, T)\),
			\begin{equation} \label{eq: 1st hedge}
				\begin{split} 
					v (t,  \Y_t ) = v (0, \Y_0) 
					&+ \int_0^t \frac{v (s,\Y_s)  - \partial^-_x v (s, \Y_s) \Y_s}{\Y^0_s} \vd  \Y^0_s 
					+ \int_0^t \partial^-_x v (s, \Y_s) \vd  \Y_s,
				\end{split} 
			\end{equation}
			where \([0, T) \ni s \mapsto v (s, S_s)\) and \([0, T) \ni s \mapsto \partial^-_x v (s, S_s)\) are locally bounded predictable processes.
			\item[\textup{(c)}] Moreover, if \(h \in C_b (\J)\), then almost surely
			\begin{equation} \label{eq: 2nd hedge}
				\begin{split}
					h (\Y_T) = E^{Q_{\Y_0}} \big[ &e^{- rT} h ( \X_T) \big] \\&+ \lim_{t \nearrow T} \Big( \int_0^{t} \frac{u (s, \Y_s) - \partial^-_x u (s, \Y_s) \Y_s}{\Y^ 0_s} \vd  \Y^0_s 
					+ \int_{0}^{t} \partial^-_x u (s, \Y_s) \vd  \Y_s \Big),
				\end{split} 
			\end{equation} 
			where it is implicit that the limit on the
			right-hand side exists.
		\end{enumerate} 
	\end{theorem}

	\begin{discussion} \label{diss: main result} 
		(i) As we discuss below in more detail, for \(h \in C_b (J^*)\), Theorem~\ref{theo: main1} establishes that, for a European contingent claim with payoff \(h(\Y_T)\),
		\begin{align} \label{eq: main hedging strategy}
			s \mapsto \Big( \frac{u (s, \Y_s) - \partial^-_x u (s, \Y_s) \Y_s}{\Y^ 0_s}, \partial^-_x u (s, \Y_s) \Big)
		\end{align}
		constitutes a hedging strategy with initial cost
		\[
		u(0,\Y_0) = E^{Q_{\Y_0}} \big[ e^{- rT} h ( \X_T) \big].
		\]
		Notably, the boundary classification of the family \((J^* \ni x \mapsto Q_x)\) is irrelevant for this result. In particular, as illustrated in Sections~\ref{ssec: 3d_Bessel} and~\ref{ssec: exit}, the case \(J \subsetneq J^*\) can occur. 
        In this case, different types of boundary behavior of \((J^* \ni x \mapsto Q_x)\)
        can yield different fundamental value functions,
        but they all yield well-defined hedges.
		To gain intuition for this observation, recall that \(\Y\) is \(J\)-valued and that the strategy in \eqref{eq: main hedging strategy} involves evaluating \(u\) and \(\partial^-_x u\) only at \(\Y_s\). Consequently, informally, points in \(J^* \setminus J\) do not affect the hedging mechanism. Rigorously, in our proof this follows from the application of a generalized It\^o-type formula to \(u (\, \cdot \,, \Y)\), which requires properties of \(u\) solely on \([0, T] \times J\) rather than on \([0, T] \times J^*\).

		\smallskip 
		(ii) Relating Theorem~\ref{theo: main1} to the classical theory, for the classical SDE
		model~\eqref{eq: SDE}, the equation~\eqref{eq: backward PDE} translates to the pricing PDE~\eqref{eq: fundamental PDE SDE model}. To be more specific, for an arbitrary reference point \(\xi \in J\), the scale function \(\s\) and speed measure \(\m\) of the SDE model \eqref{eq: SDE} are given by 
		\begin{equation} \label{eq: ScSp SDE}
			\begin{split}
				\s (x) &= \int_\xi^x \exp \Big\{ - \int^y_\xi \frac{2 b (z)}{\sigma^2 (z)} \vd z \Big\} \vd y, \quad x \in J, 
				\\
				\m (\rd x) &= \frac{\rd x}{\s' (x) \sigma^2 (x)} \quad \text{on } \mathcal{B}(J), 
			\end{split} 
		\end{equation} 
		and consequently, the auxiliary ingredients \(g\) and \(\nu\) of the PDE~\eqref{eq: backward PDE} are given by 
		\[
		\nu (\rd x) = - \frac{rx}{\sigma^2 (x)} \vd x \quad \text{on } \mathcal{B} (J), 
		\] 
		and 
		\[
		\G (x) = \exp \Big\{ - \int^x_\xi \frac{2ry}{\sigma^2 (y)} \vd y \Big\}, \quad x \in J. 
		\] 
		Now, formally, we get that 
		\[
		\frac{1}{2} \frac{\partial}{\partial \om} \frac{\partial^-u }{\partial G} (t, x)= \frac{\sigma^2 (x) \G(x)}{2} \frac{\partial}{\partial x } \frac{u_x (t, x)}{\G(x)}
		= \frac{\sigma^2 (x)}{2}
		u_{xx} (t, x) + rx \,
		u_x (t, x).
		\]
		Hence, \eqref{eq: backward PDE} reformulates to 
		\[
		u_t (t, x) + \frac{\sigma^2(x)}{2} u_{xx} (t, x) 
		+ r x \, u_x (t, x) - r u (t, x) = 0, 
		\] 
		which recovers the PDE \eqref{eq: fundamental PDE SDE model}. 
		
		\smallskip 
		(iii) In the zero interest rate regime \(r = 0\), the pricing PDE \eqref{eq: backward PDE} boils down to the Kolmogorov PDE
		\[
		\frac{\partial u}{\partial t} + \frac{1}{2} \frac{\partial}{\partial \om} \frac{\partial^-u }{\partial x} = 0, 
		\] 
		which corresponds to the diffusion on natural scale with speed measure \(\rd \om = \s'_+ \, \rd\m\). 
	\end{discussion}
	
	We proceed with a study of uniqueness. 
	The following result shows that the pricing PDE \eqref{eq: backward PDE} satisfies uniqueness among bounded continuous good solutions under the assumption that \(\J = \I\). From the perspective of mathematical finance, the latter condition will come naturally, as it is entailed by the no-arbitrage notion NFLVR, which we discuss afterwards. 
	Before we state our uniqueness result, let us recall that, by Feller's test for explosion, the condition \(\J = \I\) is equivalent to the following: 
	\begin{equation} \label{eq: no explosion} 
		\begin{split}
			b \in \{ \l, \r \}, &\ | \g (b) | < \infty
			\\& \implies	\int_I \, | \g (x)  - \g (b) |\, \om (\rd x) = \infty, \quad \forall \, I = (c, d) \subset J \text{ with } b \in \{c, d\}. 
		\end{split}
	\end{equation} 
	
	\begin{theorem} \label{theo: uniqueness pricing eq}
		Assume that \(\J = \I\) (equivalently, condition \eqref{eq: no explosion}) and \(h \in C_b (\I)\). Then, \eqref{eq: backward PDE} has only one bounded good solution \(v\) that is continuous on \([0, T] \times J\) and satisfies \(v (T, x) = h (x)\) for \(x \in \I\). This good solution is given by the fundamental value function \(u\) in~\eqref{eq:240826a1}.
	\end{theorem} 

\begin{example}\label{ex:240826a3}
We show that the assumption
$\J=\I$
in Theorem~\ref{theo: uniqueness pricing eq}
cannot be dropped.
Take any specific example with $\J\ne\I$.
Define $(\J\ni x\mapsto Q_x)$ as described above in such a way that
the boundary points in $\J\setminus\I$ are absorbing
(recall that we have some freedom to choose the boundary behavior if there are regular boundary points).
Take any $b\in\J\setminus\I$.
The function
$$
u(t,x)=0,
\quad (t,x)\in[0,T]\times J,
$$
is a bounded good solution to~\eqref{eq: backward PDE} that is continuous on
$[0,T]\times J$
and satisfies the terminal condition $u(T,x)=0$ for $x\in J$.
Define the function
$$
v(t,x)=E^{Q_x}\big[e^{-r(T-t)}\,\indic{\{b\}}(X_{T-t})\big],
\quad (t,x)\in[0,T]\times\J.
$$
By Theorem~\ref{theo: main1}~(a),
the restriction $v|_{[0,T]\times\I}$
is a bounded good solution to~\eqref{eq: backward PDE} satisfying the same terminal condition $v(T,x)=0$ for $x\in J$.
It is straightforward to verify that,
for any $[c,d]\subset\I$,
$v(t,x)\to0$, as $t\nearrow T$, uniformly in $x\in[c,d]$.
This implies that the restriction
$v|_{[0,T]\times\I}$
is continuous on $[0,T]\times\I$.
It remains to notice that
$$
u\text{ and }v|_{[0,T]\times\I}\text{ are different}.
$$
Indeed, it immediately follows from
\cite[Theorem 1.1]{bruggeman}
that $v(t,x)>0$ for all $(t,x)\in[0,T)\times\J$.
\end{example}

\begin{remark}\label{rem:240826a1}
While in Example~\ref{ex:240826a3}
the restriction $v|_{[0,T]\times\I}$
is continuous on $[0,T]\times\I$,
the function $v$ itself is clearly discontinuous
on $[0,T]\times\J$.
In Example~\ref{ex:240826a2} below we construct
in a specific situation with $\J\ne\I$ even
two different bounded good solutions $u$ and $v$
to~\eqref{eq: backward PDE}
that are both continuous on $[0,T]\times\J$
with $u(T,x)=v(T,x)$ for all $x\in\J$.
\end{remark}

	In part~(i) of Discussion~\ref{diss: main result} we already mentioned the connection of Theorem~\ref{theo: main1} to the pricing and hedging problem that motivates this paper. In particular, we pointed out that \eqref{eq: main hedging strategy} provides a hedging strategy irrespectively of the boundary behavior of \((J^* \ni x \mapsto Q_x)\).
	In the following we will discuss this connection in more detail, providing also a deterministic condition for
	the hedging strategy~\eqref{eq: main hedging strategy}
	to yield the minimal hedging capital.
	
	Set \(\oS := (\S^0, \S)\) and denote by \(L (\oS)\) the space of predictable stochastic processes $H = (H^{(0)}, H^{(1)})$
		that are \(\oS\)-integrable on all
		compact time intervals $[0,t]\subset[0,T)$.
	Furthermore, we define the \emph{value process}
	\[
	V^H := H_0 \, \oS_0 + \int_0^\cdot H_s \vd  \oS_s = H^{(0)}_0 \Y^0_0 + H^{(1)}_0 \Y_0 + \int_0^\cdot H^{(0)}_s \vd  \Y^0_s +  \int_0^\cdot H^{(1)}_s \vd  \Y_s
	\] 
	corresponding to the \emph{trading strategy} \(H = (H^{(0)}, H^{(1)}) \in L (\oS)\). 
	A trading strategy \(H = (H^{(0)}, H^{(1)})\in L (\oS)\) is called \emph{self-financing} if 
	\[
	V^H_t = H^{(0)}_t \Y^0_t + H^{(1)}_t \Y_t, \quad t \in [0, T). 
	\] 
	The set of all self-financing strategies is denoted by \(\textit{SF}\, (\oS)\). 
	Further, we define 
	\[
	\Pi_{\text{adm}} := \Big\{ H \in \textit{SF}\, (\oS) \colon \exists \, a \in \mathbb R_+ \text{ with } \inf_{t < T} V^H_t \geq - a \;\;\text{a.s.}\Big\} 
	\] 
	the set of all \emph{admissible} self-financing strategies. 
	In the context of Theorem~\ref{theo: main1}, it is interesting to relate \(E^{Q_{\Y_0}} [ e^{- r T}h(\X_T) ]\) to the so-called {\em minimal hedging capital} \(x (h, T)\) that is defined by 
	\[
	x (h, T) := \inf \Big\{ x \in \bR \colon \exists\, H \in \Pi_{\text{adm}} \text{ with } V^H_0 = x, \, \limsup_{s \nearrow T} V^H_{s} \geq h (\Y_T) \;\;\text{a.s.}\Big\}. 
	\] 
	Theorem~\ref{theo: main1} suggests the following upper bound for the minimal hedging capital of contingent claims with continuous payoff.
	\begin{lemma} \label{lem: simple ineq}
		If \(h \in C_b (J^*)\), then
		\[
		x (h, T) \leq E^{Q_{s_0}} \big[ e^{- rT} h (\X_T) \big]. 
		\] 
	\end{lemma} 
	
	\begin{proof}
		Consider the strategy
		\begin{equation} \label{eq: main TS}
			\begin{split}
				\cH_t = (\cH^{(0)}_t, \cH^{(1)}_t),\quad &\cH^{(0)}_t := \frac{u (t,\Y_t) - \partial^-_x u (t, \Y_t) \Y_t}{\Y^0_t}, 
				\quad \cH^{(1)}_t := \partial^-_x u (t, \Y_t), \quad t < T,
			\end{split} 
		\end{equation} 
		where clearly \(H \in L (\overline{S})\) due to the fact that \(u\) is a good solution by Theorem~\ref{theo: main1}~(a).
		Since Theorem~\ref{theo: main1}~(b) yields 
		\begin{align*}
			\cH^{(0)}_t \Y^0_t + \cH^{(1)}_t \Y_t = E^{Q_{s_0}} \big[ e^{- rT}h (\X_T) \big] + \int_0^t \cH^{(0)}_s \vd  \Y^0_s + \int_0^t \cH^{(1)}_s \vd  \Y_s, \quad t < T, 
		\end{align*}
		the strategy \(\cH\) is self-financing. 
		Further, again by Theorem~\ref{theo: main1}~(b), 
		\[
		V^H_{t} = u (t, \Y_t) \geq - e^{|r| T}\|h\|_\infty, \quad t < T.
		\] 
		This shows that \(H \in \Pi_{\text{adm}}\). Lastly, Theorem~\ref{theo: main1}~(c) shows that a.s.
		\[
		\limsup_{t \nearrow T} V^{\cH}_t = \lim_{t \nearrow T} V^{\cH}_t = h (\Y_T),
		\]
		and that the hedging capital \(V^{\cH}_0\) of the strategy \(\cH\) is given by \(E^{Q_{s_0}} [ e^{- rT}h (\X_T) ]\). Thus, the claim follows. 
	\end{proof}

	In the following we derive necessary and sufficient conditions for the no-arbitrage notion NFLVR or,
    equivalently,
    for the existence of an equivalent local martingale measure (ELMM),
    and we also relate the latter to the general diffusion \(Q_{s_0}\).
    With an ELMM at hand, we can use a standard martingale argument (as in \cite[Theorem on p. 710]{shir}) to get the converse inequality to that in Lemma~\ref{lem: simple ineq}. 
	
	Recall that an ELMM \(Q\) (up to the time horizon \(T\)) is a probability measure on our underlying setup, equivalent to \(P\) on \(\cF_T\), such that the discounted price process 
	\begin{equation}\label{eq:200826a1}
		\tY:= \frac{\Y}{\Y^0}
	\end{equation}
	is a local \(Q\)-martingale on the time interval \([0, T]\). Furthermore, recall at this point that all standing assumptions are in force. 
	
	\begin{theorem} \label{theo: ELMM}
		\textup{(a)} An ELMM exists if and only if \(\J = \I\) (equivalently, \eqref{eq: no explosion}) and there exists a Borel function \(\beta \colon J \to \bR\) such that 
		\[
		\int_x^y \, \beta^2 (z) \G(z) \vd z < \infty, \quad \forall \, x, y \in J, \, x < y,
		\] 
		and
			\[
			\frac{\s'_+ (y)}{\G (y)} - \frac{\s'_+ (x)}{\G (x)} = \int_x^y \beta (z) \s'_+ (z) \vd z, \quad \forall \, x, y \in J, \, x < y. 
			\]

		\textup{(b)}  If \(Q\) is an ELMM, then
		\begin{align} \label{eq: ELMM identity}
			Q \circ \Y^{-1}  = Q_{s_0} \text{ on \(\mathcal{W}_T\)}. 
		\end{align} 
	\end{theorem} 
	
	\begin{remark}
		In the zero interest rate regime \(r = 0\), we have \(\nu = 0\), \(\G (x) = 1, \g (x) = x - \xi\), and Standing Assumption~\ref{SA: nu cond} holds trivially. Thus, the two displays in Theorem~\ref{theo: ELMM}~(a) are equivalent to the structure
		\[
		\s (x) = \int^x \exp \Big\{ \int^y \beta (z) \vd z \Big\} \vd y, \quad \beta \in L^2_\textup{loc} (J), 
		\] 
		which resembles results for the existence of an ELMM from \cite[Corollary~3.11]{CU_FS_25}.
	\end{remark}

	We now obtain the following general result:
	\begin{theorem}
		\label{thm: identification}
		Assume that an ELMM exists (see Theorem~\ref{theo: ELMM}) and that \(h \in C_b (J)\). Then, the minimal hedging capital \(x (h, T)\) is given by 
		\[
		x (h, T)
        =
        u(0,s_0)
        =
        E^{Q_{s_0}} \big[ e^{- rT}h (\X_T) \big], 
		\] 
		and \(\cH\) from \eqref{eq: main TS} is a self-financing admissible hedging strategy with initial capital \(x (h, T)\). 
	\end{theorem}
	
	Summing up our observations, under the NFLVR condition (equivalently, the condition that an ELMM exists), for bounded continuous payoffs, Theorems~\ref{theo: uniqueness pricing eq}, \ref{theo: ELMM} and \ref{thm: identification} show that the pricing equation \eqref{eq: backward PDE} has a good solution \(u\), unique among bounded continuous solutions, that leads to the hedging strategy~\eqref{eq: main hedging strategy} whose initial value coincides with the minimal hedging capital.
	In case no ELMM exists, it may happen that the initial capital of~\eqref{eq: main hedging strategy} is strictly higher than the minimal hedging capital. We provide an example in Section~\ref{ssec: 3d_Bessel} below.

	\section{Market examples}
	\label{sec: examples}
	
	In this section, we put our findings into perspective by applying them to several examples of general diffusion markets with a constant interest rate.
    We always work in the setting described in the beginning of Section~\ref{subsec: market}.

	\subsection{Diffusion market with multiple skew-sticky thresholds}
	\label{ssec: skew_sticky}
	
	Our first example is a general diffusion market in which the risky asset's price follows an SDE away from finitely many thresholds, where it exhibits sticky and/or skew behavior. This example allows us to benchmark our findings against the well-established theory of classical SDE diffusion markets, such as the Black--Scholes model, which is recovered by taking the set of skew-sticky thresholds \(\mathcal Z = \emptyset\).
	
	For an open interval \(J = (\l, \r) \subset \mathbb{R}\), take two Borel functions \(b \colon J \to \mathbb{R}\) and \(\sigma \colon J \to \mathbb{R}\) that satisfy the Engelbert--Schmidt conditions
	\[
	\sigma(x)>0 \quad \forall\, x\in J, \qquad 
	\frac{1 + |b|}{\sigma^{2}} \in L^{1}_{\mathrm{loc}}(J).
	\]
	For the following we fix some $c\in J$.
	Consider a finite set \(\mc Z \subset J\),
	stickiness and skewness parameters \(\rho_\xi \geq 0\) and \(\kappa_\xi \in (0, 1)\) indexed over \(\xi \in \mc Z\), and set 
	\begin{align}
		s'_+(x) &:=  \exp \Big\{ - \int_{c}^{x} \frac{2 b(\zeta)}{\sigma^{2}(\zeta)} \vd \zeta \Big\} \begin{cases} \displaystyle \prod_{\xi \in \mc Z \cap [c, x]} \frac{1-\kappa_{\xi}}{\kappa_{\xi}}, & x \geq c, 
			\\ 
			\displaystyle \prod_{\xi \in \mc Z \cap (x, c)} \frac{\kappa_{\xi}}{1-\kappa_{\xi}}, & x < c, \end{cases} 
		\\
		s (x) &:= \int_{c}^x s'_+ (y) \vd y, 
		\\
		m(\rd x) &:= \frac{\rd x}{s'_+(x) \sigma^{2}(x)} + \sum_{\xi \in \mc Z} \frac{\rho_{\xi}}{\kappa_\xi s'_+ (\xi)}\, \delta_{\xi}(\rd x),
	\end{align}
where \(\prod_\emptyset := 1\).
	Assume further that \(s\) and \(m\) satisfy Feller's test~\eqref{eq: feller test} for the boundaries $\l $ and~$\r$.
	The pair \((s,m)\) then defines, via scale and speed, a general diffusion on \(J\) that is a unique in law weak solution to the SDE system involving local time\footnote{It is worth noting that,
		in contrast to \cite{salins2017markovprocesses},
		we always use the \emph{right} local time in this paper,
		i.e.,
		the one that is right-continuous in the space variable.}
	(see \cite{salins2017markovprocesses})
	\begin{align*}
		\mathrm d \Y_t &= b(\Y_t) \mathbf{1}_{\{\Y_t \notin \mathcal Z\}} \, \mathrm d t + \sigma(\Y_t) \1_{\{\Y_t \notin \mathcal Z\}} \, \mathrm d W_t + \sum_{\xi \in \mathcal Z} \frac{2 \kappa_\xi - 1}{2 \kappa_\xi} \, \mathrm d L^\xi_t(\Y),
		\\
		\1_{\{\Y_t = \xi\}} \, \mathrm d t &= \frac{\rho_{\xi}}{\kappa_{\xi}} \, \mathrm d L^\xi_t(\Y), \quad \forall\, \xi \in \mathcal Z.
	\end{align*}
	Define the measure $\nu $ as
	\begin{equation}
		\nu(\rd x) := - r x s'_+(x)\, m(\rd x)
		= - \frac{r x}{\sigma^{2}(x)} \vd x - \sum_{\xi \in \mc Z} \frac{r \xi \rho_{\xi}}{\kappa_\xi}\, \delta_{\xi}(\rd x)
	\end{equation}
	and notice that Standing Assumption~\ref{SA: nu cond} means
	\begin{equation}\label{eq:220826a1}
		2 r \xi \rho_\xi < \kappa_\xi,
		\quad \forall\, \xi \in \mathcal Z.
	\end{equation}
	We set
	\begin{equation}\label{eq:220826a2}
		\wt\kappa_\xi
		:=
		\frac{\kappa_\xi}{2(\kappa_\xi- r\xi\rho_\xi)},
		\quad
		\xi\in\mc Z,
	\end{equation}
	and observe that $\wt\kappa_\xi\in(0,1)$ due to~\eqref{eq:220826a1}.
	With notation~\eqref{eq:220826a2}, the characteristics of the auxiliary diffusion $(G,\ol m) $ are
	\begin{align}
		g(x) &:=  \exp\Big\{- \int_{c}^{x} \frac{2 r\zeta}{\sigma^{2}(\zeta)}\vd \zeta \Big\}
		\begin{cases} \displaystyle \prod_{\xi \in \mc Z \cap [c, x]} \frac{1-\wt\kappa_\xi}{\wt\kappa_\xi}, & x \geq c, 
			\\
			\displaystyle \prod_{\xi \in \mc Z \cap (x, c)} \frac{\wt\kappa_\xi}{1-\wt\kappa_\xi}, & x < c,
		\end{cases}
		\\
		G (x) &:= \int_{c}^x g (y) \vd y, 
		\\
		\ol m(\rd x) &:= \frac{\rd x}{G'_+(x) \sigma^{2}(x)} + \sum_{\xi \in \mc Z} 
		\frac{\rho_{\xi} }{\kappa_\xi G'_+(\xi) } \delta_\xi (\rd x) ,
	\end{align}
	and the hedging PDE~\eqref{eq: backward PDE} takes the form
	\begin{equation}
		\label{eq: skew_sticky:hedging_equation}
		\begin{dcases}
			u_t(t,x) + \tfrac{1}{2} \sigma^2(x) u_{xx}(t,x) + r x \, u_x(t,x) - r u(t,x) = 0, & x \in J \setminus \mathcal Z, \, t \in (0,T), 
			\\[3pt]
			\wt \kappa_{\xi} \partial^-_x u(t, \xi+)
			-
			(1 - \wt \kappa_{\xi}) \partial^-_x u(t, \xi)
			=
			\frac{2 \rho_{\xi} \wt \kappa_\xi}{\kappa_\xi}
			\big( r u(t, \xi) - u_t(t, \xi) \big),
			& t \in (0,T),\, \xi \in \mathcal Z, 
			\\[3pt]
			u(T,x) = h(x), & x \in J.
		\end{dcases}
	\end{equation}
	Applying Theorem~\ref{theo: ELMM} to this model, an ELMM exists if and only if
	\begin{equation}
		\label{eq: ELMM_SkS}
		\eqref{eq: no explosion} \text{ holds,}\quad
		\kappa_\xi=\wt\kappa_\xi,
		\ \forall\, \xi \in \mathcal Z,
		\quad \text{and} \quad 
		\Big( x \mapsto \frac{ r x - b(x) }{\sigma^2 (x)} \Big) \in L^2_{\mathrm{loc}}(J).
	\end{equation}
Some remarks concerning these conditions are in order.
		
		\begin{remark}\label{rem:220826a2}
			(a)
				In terms of our input parameters only,
				the second condition in~\eqref{eq: ELMM_SkS} is equivalently expressed as
				$$
				2 r \xi \rho_{\xi} = 2\kappa_{\xi} - 1,
				\quad \forall\, \xi \in \mathcal Z.
				$$
				It is worth noting that the latter relation
				implies~\eqref{eq:220826a1}, as $\kappa_\xi<1$.
			
			\smallskip
			\noindent
			(b) The non-explosion condition \eqref{eq: no explosion} is satisfied if and only if, for an arbitrary \(c \in J\), 
			\begin{align*}
				\int_{\l +}^c e^{\int_x^c \frac{2rz}{\sigma^2(z)} \vd z} \Big( \int_{x}^c e^{- \int_y^c \frac{2rz}{\sigma^2(z)} \vd z} \frac{\vd y}{\sigma^2(y)} \Big) \vd x &= \infty, \\
				\int_c^{\r -} e^{- \int_c^x \frac{2rz}{\sigma^2(z)} \vd z} \Big(\int_c^x e^{\int_c^y \frac{2rz}{\sigma^2(z)} \vd z} \frac{\vd y}{\sigma^2 (y)} \Big) \vd x &= \infty, 
			\end{align*}
			see \cite[Eq. (5.65) on p. 347, Theorem~5.5.29]{KaraShre}. For the canonical choices \(J = \bR, J = (0, \infty)\), and \(J = (0, 1)\), a straightforward analysis of the above integral test yields the following conditions:
			\begin{enumerate}
				\item[(b.1)] In case \(J = \bR\), \eqref{eq: no explosion} is always satisfied. This observation is comparable to the well-known fact that SDEs without drift cannot explode out of the real line, see \cite[Problem~5.5.3]{KaraShre}.
				\item[(b.2)] In case \(J = (0, \infty)\), \eqref{eq: no explosion} is equivalent to 
				\begin{align} \label{eq: no 0 explosion}
					\int_{0 + }^c \frac{x}{\sigma^2 (x)} \vd x = \infty.
				\end{align} 
				Again, no additional constraint at the infinite boundary is needed.
				\item[(b.3)] In case \(J = (0, 1)\), one needs to distinguish between different values of the interest rate \(r\). More specifically, \eqref{eq: no explosion} is satisfied if and only if \eqref{eq: no 0 explosion} and one of the following two conditions hold:
				
				\noindent
				(b.3.1) \(r = 0\) and 
				\[
				\int^{1-}_c \frac{1 - x}{\sigma^2 (x)} \vd x = \infty; 
				\]
				
				\noindent
				(b.3.2) \(r < 0\) and 
				\[
				\int^{1-}_c e^{- 2r \int^x_c \frac{z}{\sigma^2 (z)} \vd z} \vd x = \infty.
				\]
				In particular, \eqref{eq: no explosion} is always violated when \(r > 0\), which intuitively means that the positive drift pushes the process towards the boundary \(1\), irrespectively of the diffusive behavior. 
			\end{enumerate}
			(c) 			As mentioned in Remark~\ref{rem:220826a1}, the NIP no-arbitrage condition implies \eqref{eq: (ii) from weak arbitrage paper}, which means in the current setting that \(2 r \xi \rho_\xi = 2 \kappa_\xi - 1\) for all \(\xi \in \mathcal{Z}\)
			(equivalently,
				$\kappa_\xi=\wt\kappa_\xi$
				for all \(\xi \in \mathcal{Z}\);
				see item~(a)).
			As a consequence, if this condition fails, the market even allows for increasing profits, which may be constructed explicitly as described in \cite{ACU_25_SIFIN}.
	\end{remark} 
	Returning to the pricing and hedging problem, by Theorem~\ref{thm: identification}, if the condiditions from~\eqref{eq: ELMM_SkS} are satisfied, the hedging PDE is the fundamental pricing PDE, which reads:
	\begin{equation}
		\label{eq: skew_sticky:hedging_equation_RN}
		\begin{dcases} 
			u_t(t,x) + \tfrac{1}{2} \sigma^2(x) u_{xx}(t,x) + r x \, u_x(t,x) - r u(t,x) = 0, & x \in J \setminus \mathcal Z, \, t \in (0,T),
			\\[3pt]
			\kappa_{\xi} \partial^-_x u(t, \xi+)
			-
			(1 - \kappa_{\xi}) \partial^-_x u(t, \xi)
			=
			2 \rho_{\xi}
			\big( r u(t, \xi) - u_t(t, \xi) \big),
			& t \in (0,T),\, \xi \in \mathcal Z,
			\\[3pt]
			u(T,x) = h(x), & x \in J,
		\end{dcases}
	\end{equation}
	and the initial value of the unique good solution is the minimal hedging capital of the contingent claim with payoff~\(h\).
	
	\begin{remark}
	If the ELMM exists, comparing the initial price process and the auxiliary diffusion (implied from~\eqref{eq: skew_sticky:hedging_equation_RN}),
		we observe that only the drift is modified to match the growth of the risk-free asset; all other parameters (skewness and stickiness $(\kappa_{\xi},\rho_{\xi}) $ at every threshold and the diffusion coefficient~$\sigma $) are left unchanged.
	\end{remark}

	\subsection{Diffusion market where $(G,\ol m) $ has a regular boundary}
	\label{ssec: 3d_Bessel}
Here we revisit a famous example of Delbaen and Schachermayer
\cite{DS1995a}
for the existence of arbitrage (with admissible strategies).
It is further investigated e.g.
in \cite[Example 4.6]{karatzaskardaras07}
and
in \cite[Section 6]{R_13}.
We now discuss the example from the perspective of our hedging PDE.
As usual in the setting of this paper, we consider an interest rate $r\in\mathbb R$, which will appear in formulas below.

	Consider the diffusion with state space $\I=(0,\infty)$ defined via the scale and speed
	\begin{equation}
		s(x) = - \frac{1}{x}, \quad m(\rd x) = x^{2} \vd x, \quad x>0.
	\end{equation}
	These characteristics \((\s, \m)\) correspond to a three-dimensional Bessel process (BES$^3$), which has an inaccessible boundary at $0$ (more precisely, a so-called entrance boundary), and whose dynamics is captured by the SDE
	\begin{equation}
		\vd \Y_t = \frac{1}{\Y_t} \vd t + \rd W_t,
	\end{equation}
	where $W$ is a Brownian motion. For an exhaustive account of the BES$^3$ process, see \cite[Section VI.3]{RY}. In this case, the characteristics $(G,\ol m)$ of the auxiliary diffusion are of the form
	\begin{equation}
		G(x) = \int_0^x g(y) \vd y, \quad
		\ol m(\rd x) =  \frac{\rd x}{g(x)} ,\quad 
		g(x) = \exp \big\{-r x^2\big\}.
	\end{equation}
	That is, the generator on \((0, \infty)\) is of the Ornstein--Uhlenbeck-type, the state space is $\J=[0,\infty)$, while $0$ is a regular boundary point.
    In particular, $\J\ne\I$, hence, by Theorem~\ref{theo: ELMM}, NFLVR fails.
	On the other hand, the process
	\[
	Z_t := \exp \Big\{ - \int_0^t \Big( \frac{1}{\Y_s} - r \Y_s \Big) \vd W_s - \frac{1}{2} \int_0^t \Big( \frac{1}{\Y_s} - r \Y_s \Big)^2 \vd s  \Big\}, \quad t \in [0, T],
	\]
	is a so-called \emph{strict martingale density (SMD)},
	i.e., a strictly positive local martingale such that \(Z \tY\) is a local martingale
	(for the notation $\tY$, recall~\eqref{eq:200826a1}).
    This means that the weaker than NFLVR no-arbitrage condition of ``no unbounded profit with bounded risk'' (NUPBR) holds.
    For more details on NUPBR, in particular, for many characterizations in continuous market models, we refer to
    \cite{Fontana2015,HS10}.
    Moreover, in this setting, $Z$ is a strict local martingale, i.e., a local martingale but not a true martingale
    (because otherwise the measure with the density process $Z$ would be an ELMM).
    In particular, \(E^P [ Z_T ] < 1\).

    We proceed with two specific examples in the described setting.

\begin{example}\label{ex:240826a1}
We illustrate that without NFLVR our hedge does not necessarily achieve the minimal hedging capital,
i.e.,
the assumption of the existence of an ELMM in Theorem~\ref{thm: identification} cannot be dropped.
Let $(\cF^\Y_t)_{t\ge0}$ denote the right-continuous filtration generated by $\Y$.
We first observe that $Z$ is an SMD also on the filtration
$(\cF^\Y_t)_{t \geq 0}$.
Further, on $(\cF^\Y_t)_{t \geq 0}$, the discounted price process
$\tY$
possesses the predictable representation property,
as it is inherited from that of~$\Y$
(see \cite[Theorem~2.1]{CU24} for the latter).
By an application of Theorems 3.46 and~3.38 in \cite{karatzaskardaras},
$Z$ is the only SMD on $(\cF^\Y_t)_{t \geq 0}$.
Now, consider an arbitrary SMD $\wt Z$ on $(\cF_t)_{t \geq 0}$.
It follows from \cite[Theorem 3.1]{KardarasRuf2020} that a.s.
$E^P[\wt Z_T \mid \cF^\Y_T]\le Z_T$.
This, in particular, implies that, for any nonnegative claim
$h\colon\J\to[0,\infty)$,
we have
$$
E^P\big[\wt Z_T e^{-rT} h(\Y_T)\big]\le E^P[Z_T e^{-rT} h(\Y_T)].
$$
Then, \cite[Theorem~3.7]{karatzaskardaras} yields that the claim $h(\Y_T)$ can be hedged
from the starting capital
$E^P[Z_T e^{-rT} h(\Y_T)]$
via a self-financing trading strategy with a nonnegative value process
provided $E^P[Z_T e^{-rT} h(\Y_T)]<\infty$.\footnote{This seemingly overcomplicated argument is due to the fact that we allow for incompleteness, i.e.,
the underlying filtration $(\cF_t)_{t \geq 0}$ can be bigger than $(\cF^\Y_t)_{t \geq 0}$.}
Thus, it now remains to take a constant function
$h\colon\J\to\mathbb R$, $h(x)\equiv h>0$,
and to observe that
$$
E^P [ Z_T e^{- rT} h (\Y_T) ]
=
e^{- rT}h \, E^P [Z_T]
<
e^{- rT}h
=
E^{Q_{s_0}} [ e^{- rT} h (X_T) ]
=
u(0,s_0),
$$
which shows that, in this example, our fundamental value function leads to a higher price than necessary.
\end{example}

\begin{example}\label{ex:240826a2}
We now strengthen the message of Example~\ref{ex:240826a3} as discussed in Remark~\ref{rem:240826a1},
that is,
we construct two different bounded continuous functions
$u^{(i)}\colon[0,T]\times\J\to\mathbb R$,
$i\in\{1,2\}$,
such that they are both good solutions
to~\eqref{eq: backward PDE}
with $u^{(1)}(T,x)=u^{(2)}(T,x)$ for all $x\in\J$.

Recall that, according to the general methodology explained in Section~\ref{sec: hedging_PDE},
we need to extend $\om$ to $\J=[0,\infty)$ in an arbitrary way.
To be precise, take $m_0\in[0,\infty]$ and set
$$
\om(\{0\}):=m_0.
$$
This makes the boundary point $0$ absorbing
(resp., instantaneously reflecting;
resp., slowly reflecting)
for the auxiliary diffusion
$(\J\ni x\mapsto Q_x)$
if $m_0=\infty$
(resp., $m_0=0$;
resp., $m_0\in(0,\infty)$).
Take a continuous strictly decreasing function
$h\colon\J=[0,\infty)\to(0,\infty)$.
In particular, $h\in C_b(\J)$.
Below in this example we write
$u^{(1)}$ and $Q_x^{(1)}$
(resp., $u^{(2)}$ and $Q_x^{(2)}$)
for the fundamental value function~\eqref{eq:240826a1}
defined on $[0,T]\times\J$
and for the diffusion measure on the path space
that correspond to $m_0=\infty$
(resp., to some fixed $m_0\in[0,\infty)$).
By Theorem~\ref{theo: main1}, for $i\in\{1,2\}$,
the function $u^{(i)}$ is a bounded good solution to~\eqref{eq: backward PDE} satisfying
$u^{(i)}(T,x)=h(x)$ for $x\in\J$.
The continuity of $u^{(i)}$ on $[0,T]\times\J$
follows from \cite[Lemma~30, p.~116]{freedman}.
Using~\eqref{eq:240826a1},
\cite[Theorem 1.1]{bruggeman}
and the fact that,
under $Q_x^{(2)}$,
the coordinate process $X$ stopped upon hitting $0$
has distribution $Q_x^{(1)}$,
we obtain that
$u^{(1)}(t,x)>u^{(2)}(t,x)$ for all $(t,x)\in[0,T)\times\J$.
This proves the claimed non-uniqueness.
\end{example}

	\subsection{Diffusion market where $(G,\ol m) $ has an exit boundary}
	\label{ssec: exit}

    In the spirit of Section~\ref{ssec: 3d_Bessel}, we now present an example where the auxiliary diffusion with characteristics \((\g, \ol m)\) has an exit boundary (in contrast to the regular boundary case from Section~\ref{ssec: 3d_Bessel}).
    For the financial part of the story, this again means that NFLVR fails.
    For purely probabilistic part of the story,
    this shows that all accessible boundary classifications are possible.
    Let the diffusion $\Y$ with state space $\I=(0,\infty)$
    be driven by the SDE
    $$
    \rd \Y_t
    =
    2 \vd t
    +
    2 \sqrt{\Y_t} \vd W_t,
    $$
    where $W$ is a Brownian motion.
    In other words, $\Y$ is a two-dimensional squared Bessel process (BESQ$^2$).
    The boundary point $0$ is inaccessible for BESQ$^2$
    (more precisely, entrance boundary).
    For completeness, the scale and speed are
    $$
    s(x) = \log x,
    \quad
    m(\rd x) = \frac14 \vd x,
    \quad
    x>0.
    $$
    For an extensive account of BESQ processes,
    see \cite[Section XI.1]{RY}.

	In this case, the auxiliary characteristics $(G,\ol m) $ take the form
	$$
		G(x) = \int_0^x e^{-r y / 2} \vd y,
        \quad
		\ol m(\rd x) =  \frac1{4x} \, e^{r x / 2} \vd x,
        \quad
        x>0
    $$
    (we leave $G$ in the integral form in order to avoid distinguishing between the cases $r=0$ and $r\ne0$).
    By Feller's test, the boundary point $0$ is accessible,
    that is, the state space is $\J=[0,\infty)$.
    The fact that $\om((0,\varepsilon))=\infty$ for all $\varepsilon>0$ yields that $0$ is an exit boundary of the auxiliary diffusion
    (or see, e.g., \cite[Section 5.11]{ito2006book} for integral criteria).
	For completeness, we observe that the auxiliary diffusion here is driven by the SDE
	$$
	\rd Y_t
    = 
	r Y_t \vd t
    +
    2 \sqrt{Y_t} \vd B_t,
    $$
    where $B$ is a Brownian motion
	(the notation $Y,B$ instead of $\Y,W$ is to emphasize that the auxiliary diffusion as a process need not be related to the original diffusion as a process).
    To provide also the SDE perspective, we remark that the fact $0$ is an exit boundary point for $Y$ is comparable to the well-known fact that $0$ is an exit boundary for BESQ$^0$ SDE
    $\rd Y_t=2\sqrt{Y_t}\vd B_t$.
    In contrast to Section~\ref{ssec: 3d_Bessel}, here the auxiliary diffusion cannot be chosen to have an instantaneous or slow reflection at~$0$.

	\section{Numerical experiments}
	\label{sec: numexp}
	
	In this section, we provide numerical illustrations of the main theoretical results from Sections~\ref{sec: hedging_PDE} and~\ref{ssec: skew_sticky}. The experiments serve three purposes: (i)~to demonstrate that the discrete hedging strategy derived from Theorem~\ref{theo: main1} indeed yields vanishing tracking error; (ii)~to show that the strategy remains effective even when the NFLVR condition fails, though the associated hedging capital may be non-minimal; and (iii)~to provide a concrete example of the general diffusion framework and its hedging PDE.
	
	Since continuous hedging is not feasible in practice, we implement discretized hedging strategies with a finite number of portfolio rebalancing instances. In smooth SDE diffusion models, the tracking error is known to converge in \(L^2\) to \(0\) at rate \(O(N^{-1/2})\), where \(N\) is the number of uniform-in-time rebalancing instances (see the seminal work of Bertsimas et al.~\cite{bertsimas2000when}). The solution to the hedging equation and delta field are approximated numerically using an Implicit Finite Difference (IFD) scheme on the PDE problem~\eqref{eq: backward PDE}, while the hedging performance is evaluated using Monte Carlo estimation over a number of simulated path scenarios. 
	The emphasis of this section is on the qualitative behaviour predicted by the theory, not on a rigorous convergence analysis of the numerical methods themselves.
	
	\subsection{Model Configuration}
	
	We consider a two‐threshold general diffusion on \(\mathbb{R}\) with scale density and speed measure given by
	\begin{equation}
    \label{eq: SKS_sm}
		\begin{aligned}
		s'_+(x) &:= \frac{\kappa_{-1}}{1-\kappa_{-1}}\mathbf{1}_{(-\infty,-1)}(x) + \mathbf{1}_{[-1,1)}(x) + \frac{1-\kappa_1}{\kappa_1}\mathbf{1}_{[1,\infty)}(x),\\
		s(x) &:= \int_0^x s'_+(y)\,dy,\\
		m(dx) &:= \frac{dx}{s'_+(x)} + \frac{\rho_{-1}}{\kappa_{-1}}\delta_{-1}(dx) + \frac{\rho_1}{1-\kappa_1}\delta_1(dx),
	\end{aligned}
	\end{equation}
	with parameters \(\kappa_{-1},\kappa_1\in(0,1)\), \(\rho_{-1},\rho_1\ge 0\), and constant interest rate \(r\in\mathbb{R}\). 
	This family includes the Bachelier model when \((\kappa_{-1},\kappa_1,\rho_{-1},\rho_1)=(0.5,0.5,0,0)\). 
	Existence of an ELMM in this model depends solely on whether the balance condition is satisfied at \(\pm 1\):
	\[
	-2r\rho_{-1}=2\kappa_{-1}-1,\qquad 2r\rho_1=2\kappa_1-1.
	\]
	
	We price and hedge a European bear spread option with payoff
	\[
	h(x) = (K_{\mathrm{high}} - x)^+ - (K_{\mathrm{low}} - x)^+,\qquad a^+:=\max\{a,0\},
	\]
	with lower strike \(K_{\mathrm{low}}=-2\), upper strike \(K_{\mathrm{high}}=2\), and time horizon \(T=10\). This bounded payoff is consistent with our \(C_b\) setting and features two kinks at the strike levels, which provide a nontrivial hedging experiment. 
	
	\subsection{Numerical Pipeline}

Instead of looking directly at the hedging PDE, we focus on the PDE solved by $v(t,x):= e^{r(T-t)}u(t,x)$, which is
\begin{align} \label{eq: forward PDE}
	\frac{\partial v}{\partial t} + \frac{1}{2} \frac{\partial}{\partial \om} \frac{\partial^- v }{\partial \g} = 0 \quad \text{ on }
	{\color{magenta}[}0, T) \times J, 		
\end{align} 
Our choice is justified by the empirical observation that the finite difference scheme for this equation produces less biased value fields than for the hedging PDE~\eqref{eq: backward PDE}. 

For the general diffusion market defined by~\eqref{eq: SKS_sm}, the PDE~\eqref{eq: forward PDE} reads
\begin{equation}
	\label{eq: forward PDE_spec}
	\begin{dcases} 
		v_t(t,x) + \tfrac{1}{2} v_{xx}(t,x) + r x \, v_x(t,x) = 0, & x \in J \setminus \{-1,1\}, \, t \in (0,T),
		\\
		\wt \kappa_{-1} \partial^-_x v(t, (-1)+)
		-
		(1 - \wt \kappa_{-1}) \partial^-_x v(t, -1)
		=
		- \frac{2 \rho_{-1} \wt \kappa_{-1}}{\kappa_{-1}}
		 v_t(t, -1),
		& t \in (0,T),
		\\
		\wt \kappa_{1} \partial^-_x v(t, 1+)
-
(1 - \wt \kappa_{\xi}) \partial^-_x v(t, 1)
=
- \frac{2 \rho_{1} \wt \kappa_1}{\kappa_1}
v_t(t, 1),
		& t \in (0,T),
		\\
		u(T,x) = h(x), & x \in \IR,
	\end{dcases}
\end{equation} 
where $\wt \kappa_{\xi} $ is given by~\eqref{eq:220826a2}. 

We approximate the solution to~\eqref{eq: forward PDE_spec} on a truncated domain \([-50,50]\) using a fully implicit time-stepping method. At the artificial boundaries \(x = \pm 50\), we impose homogeneous Neumann conditions \(\partial_x u = 0\), which is natural since the payoff function satisfies \(\partial_x h = 0\) at \(\pm 50\). The domain is chosen sufficiently large so that boundary effects are negligible. The jump conditions at the skew-sticky interfaces \(\pm 1\) are imposed by solving the discretized version of \eqref{eq: skew_sticky:hedging_equation} at the relevant grid points.

We discretize the domain with a uniform grid in time of $N^{\mathrm{time}}_{\mathrm{FD}} =500 $ points and a grid in space on $[-50,50] $ of $N^{\mathrm{space}}_{\mathrm{FD}} \in \{1000,2000,4000\} $ points that is:
\begin{itemize}
	\item of uniform step $h_{\mathrm{step}} = 100/ N^{\mathrm{space}}_{\mathrm{FD}} $ away from the skew-sticky points
	
	\item of step $h_{\mathrm{step}}^{2} $ around those points.
\end{itemize}
The left-hand spacial derivative approximation is obtained by finite differences:
\begin{equation}
	\partial^{-}_{x}v(t_i,x_j) \approx \frac{v(t_{i},x_{j}) - v(t_{i},x_{j-1})}{x_j - x_{j-1}}.
\end{equation}
Intermediate values of these quantities are obtained via linear interpolation between grid points, with constant extrapolation near the interfaces to respect the jump in delta. 
For any $(t,x)\in [0,T]\times \IR $, let $\mathfrak v(t,x) $, $\partial^{-}_{x} \mathfrak v(t,x) $ be the evaluations of these interpolations.
By the relation between the solutions of~\eqref{eq: backward PDE} and  \eqref{eq: forward PDE}, we infer the following numerical approximations of the solution to the true hedging equation $u $ and its associated  delta field $\partial^{-}_xu $:
\begin{align}
	(t,x) &\mapsto \texttt{value\_field}(t,x) := e^{-r(T-t)} \mathfrak v(t,x), \\
	 (t,x) &\mapsto \texttt{delta\_field}(t,x) := e^{-r(T-t)} \partial^{-}_{x} \mathfrak v(t,x) ,\quad t\in [0,T], \quad x\in [-50,50].
\end{align}

	To generate paths of the diffusion, we use the Space‐Time Markov Chain Approximation from \cite{anagnostakis2023general} with \(N_{\mathrm{MC}}=2000\) paths. For each path, we apply the discrete delta‐hedging strategy (Algorithm~\ref{algo_deltahedging}) for numbers of rebalancing instances \(N\in\{32,64,\dots,4096\}\) using the IFD approximations of the price and delta fields. 
	The IFD scheme and the path simulation are standard tools and are used here to obtain approximate price and delta fields; their accuracy is sufficient to illustrate the qualitative behaviour predicted by the theory. A brief validation against the Bachelier model 
	(for the case without skew-sticky points) confirms that the numerical errors are small relative to the effects under study.

	\begin{table}[htbp]
		\caption{Hedging performance metrics (MTE* and StDTE*) for all models, where the $\pm $ indicates the 95\% confidence interval.
		Results are for a finite difference scheme with  $N^{\mathrm{time}}_{\mathrm{FD}} =500 $ steps in time,  $N^{\mathrm{space}}_{\mathrm{FD}} = 4000 $ steps in space, and a Monte Carlo sample size of $N_{\mathrm{MC}} = 2000 $. All values are rounded to three decimals; values of 0.000 indicate a magnitude smaller than 0.0005.}
		\label{tab:results}
		\centering
		\begin{subtable}{0.5\textwidth}
			\centering
			\subcaption{Bachelier (Premium $= 2.0$)}
			\label{tab:bachelier}
			\begin{tabular}{r |r |r}
				\toprule
				$N$ & MTE* & StDTE* \\
				\midrule
				32 & $-0.002\,  \pm0.008$ & $0.192\,  \pm0.009$ \\
				64 & $-0.001\,  \pm0.006$ & $0.137\,  \pm0.007$ \\
				128 & $-0.000\,  \pm0.004$ & $0.100\,  \pm0.005$ \\
				256 & $-0.002\,  \pm0.003$ & $0.070\,  \pm0.003$ \\
				512 & $-0.000\,  \pm0.002$ & $0.050\,  \pm0.002$ \\
				1024 & $-0.001\,  \pm0.002$ & $0.036\,  \pm0.002$ \\
				2048 & $-0.001\,  \pm0.001$ & $0.027\,  \pm0.001$ \\
				4096 & $-0.001\,  \pm0.001$ & $0.019\,  \pm0.001$ \\
				\bottomrule
			\end{tabular}
		\end{subtable}
		~\;~
		\begin{subtable}{0.5\textwidth}
			\centering
			\subcaption{Skew-Sticky 1 (Premium $= 0.271$)}
			\label{tab:ss1}
			\begin{tabular}{r |r |r}
				\toprule
				$N$  & MTE* & StDTE* \\
				\midrule
				32 & $0.058\,  \pm0.022$ & $0.506\,  \pm0.024$ \\
				64 & $0.045\,  \pm0.018$ & $0.408\,  \pm0.019$ \\
				128 & $0.032\,  \pm0.015$ & $0.337\,  \pm0.016$ \\
				256 & $0.022\,  \pm0.012$ & $0.269\,  \pm0.013$ \\
				512 & $0.015\,  \pm0.010$ & $0.223\,  \pm0.011$ \\
				1024 & $0.008\,  \pm0.008$ & $0.184\,  \pm0.009$ \\
				2048 & $0.009\,  \pm0.007$ & $0.156\,  \pm0.007$ \\
				4096 & $0.013\,  \pm0.006$ & $0.131\,  \pm0.006$ \\
				\bottomrule
			\end{tabular}
		\end{subtable}
		\\[0.2in]
		\begin{subtable}{0.5\textwidth} 
			\centering
			\subcaption{Skew-Sticky 2 (Premium $=14.778
				$)}
			\label{tab:ss2}
			\begin{tabular}{r |r |r}
				\toprule
				$N$  & MTE* & StDTE* \\
				\midrule
				32 & $-0.007\,  \pm0.005$ & $0.115\,  \pm0.005$ \\
				64 & $-0.004\,  \pm0.004$ & $0.080\,  \pm0.004$ \\
				128 & $-0.003\,  \pm0.003$ & $0.058\,  \pm0.003$ \\
				256 & $-0.002\,  \pm0.002$ & $0.042\,  \pm0.002$ \\
				512 & $-0.002\,  \pm0.001$ & $0.032\,  \pm0.002$ \\
				1024 & $-0.001\,  \pm0.001$ & $0.025\,  \pm0.001$ \\
				2048 & $-0.001\,  \pm0.001$ & $0.019\,  \pm0.001$ \\
				4096 & $-0.001\,  \pm0.001$ & $0.016\,  \pm0.001$ \\
				\bottomrule
			\end{tabular}
		\end{subtable}
		~\;~
		\begin{subtable}{0.5\textwidth} 
			\centering
			\subcaption{Skew-Sticky 3 (Premium $= 2.0$)}
			\label{tab:ss3}
			\begin{tabular}{r |r |r}
				\toprule
				$N$  & MTE* & StDTE* \\
				\midrule
				32 & $0.005\,  \pm0.009$ & $0.201\,  \pm0.010$ \\
				64 & $0.002\,  \pm0.007$ & $0.151\,  \pm0.007$ \\
				128 & $0.005\,  \pm0.005$ & $0.118\,  \pm0.006$ \\
				256 & $0.003\,  \pm0.004$ & $0.093\,  \pm0.004$ \\
				512 & $0.000\,  \pm0.003$ & $0.074\,  \pm0.003$ \\
				1024 & $0.000\,  \pm0.003$ & $0.059\,  \pm0.003$ \\
				2048 & $0.001\,  \pm0.002$ & $0.048\,  \pm0.002$ \\
				4096 & $0.001\,  \pm0.002$ & $0.041\,  \pm0.002$ \\
				\bottomrule
			\end{tabular}
		\end{subtable}
	\end{table}

The tracking error for a single path and hedging frequency \(N\) is the value of the replication portfolio plus a short position on the option, defined as
\[
\epsilon_T^N(S) :=  \texttt{value\_field}(S_0, T) + \int_0^T H_s^{(0,N)}\,dS_s^0 + \int_0^T H_s^{(1,N)}\,dS_s - h(S_T),
\]
where \((H^{(0,N)},H^{(1,N)})\) are the predictable step-processes defined as
\begin{equation}
	H^{(0,N)} := \sum_{i=0}^{N-1} \indicB{t\in [i\frac{T}{N},(i+1)\frac{T}{N})} \wh H^{(0)}_{i\frac{T}{N}},\quad 
	H^{(1,N)} :=  \sum_{i=0}^{N-1} \indicB{t\in [i\frac{T}{N},(i+1)\frac{T}{N})} H^{(1)}_{i\frac{T}{N}},\quad t\in [0,T],
\end{equation} 
where $H^{(1)} $ is the continuous strategies from Theorem~\ref{theo: main1}, written explicitly in~\eqref{eq: main TS}, and $\wh H^{(0)} $ is the strategy that satisfies $\wh H^{(0)}_0 = H^{(0)}_0 $ ($H^{(0)}_0 $ is the cash at $t=0 $ in the continuous hedging strategy, see~\eqref{eq: main TS}) and the self-financing constrain
\[
\wh H_{t_i}^{(0)} = \frac{\wh H_{t_{i-1}}^{(0)} S_{t_i}^{0} + \left( H^{(1)}_{t_{i-1}} - H^{(1)}_{t_i} \right) S_{t_i} }{ S_{t_i}^{0} }, \quad i = 0,1,\dots,N.
\]
    
    We report Monte Carlo estimates
    of the mean tracking error (MTE) and of the standard deviation of the tracking error (StDTE)
    over all simulated scenarios:
	\[
	\mathrm{MTE}^* := \frac{1}{N_{\mathrm{MC}}}\sum_{i=1}^{N_{\mathrm{MC}}}\epsilon_T^N(S^i),\qquad
	\mathrm{StDTE}^* := \sqrt{\frac{1}{N_{\mathrm{MC}}}\sum_{i=1}^{N_{\mathrm{MC}}}(\epsilon_T^N(S^i) - \mathrm{MTE}^*)^2}.
	\]
	
	While the root mean squared tracking error (RMSTE), defined by \(\sqrt{\mathbb{E}[(\epsilon_T^N(S))^2]}\), is the standard scalar metric for the tracking error (see \cite{bertsimas2000when}), we report MTE and StDTE separately as they provide more insight: MTE reveals systematic bias, while StDTE quantifies the hedging uncertainty around it. Moreover, RMSTE can be recovered from these two quantities via the relation
	\[
	\mathrm{RMSTE}^* := \sqrt{\frac{1}{N_{\mathrm{MC}}}\sum_{i=1}^{N_{\mathrm{MC}}} (\epsilon_T^N(S^i))^2} = \sqrt{(\mathrm{MTE}^*)^2 + (\mathrm{StDTE}^*)^2}.
	\]
	
	Asymptotic 95\% confidence intervals for the MTE are computed as \(\mathrm{MTE}^* \pm 1.96 \cdot \mathrm{StDTE}^* / \sqrt{N_{\mathrm{MC}}}\), while confidence intervals for the StDTE are obtained using a chi-squared approximation. These are used only to indicate sampling variability and are not intended as a rigorous statistical analysis.

	\subsection{Model Parameters and ELMM Status}
	
	We consider four configurations, summarised in Table~\ref{table:parameters}. The first three satisfy the balance condition (ELMM exists); the fourth violates it, providing a test case where the hedging framework of Theorem~\ref{theo: main1} is applied without an ELMM.
	
	\subsection{Results and Discussion}
	
	\begin{table}[htbp]
		\centering
		\caption{Model parameters and ELMM status.}
		\label{table:parameters}
		\begin{tabular}{|c|c|c|c|c|c|c|}
			\hline
			\textbf{Model} & \(\kappa_{-1}\) & \(\kappa_1\) & \(\rho_{-1}\) & \(\rho_1\) & \(r\) & \textbf{ELMM exists?} \\
			\hline
			Bachelier & 0.5 & 0.5 & 0 & 0 & 0 & Yes (trivial) \\
			Skew-Sticky 1 & 0.3 & 0.7 & 1 & 1 & 0.2 & Yes \\
			Skew-Sticky 2 & 0.7 & 0.3 & 1 & 1 & -0.2 & Yes \\
			Skew-Sticky 3 & 0.7 & 0.7 & 1 & 1 & 0 & No \\
			\hline
		\end{tabular}
	\end{table}
	
	\begin{figure}[htbp]
		\centering
		\includegraphics[width=0.5\textwidth]{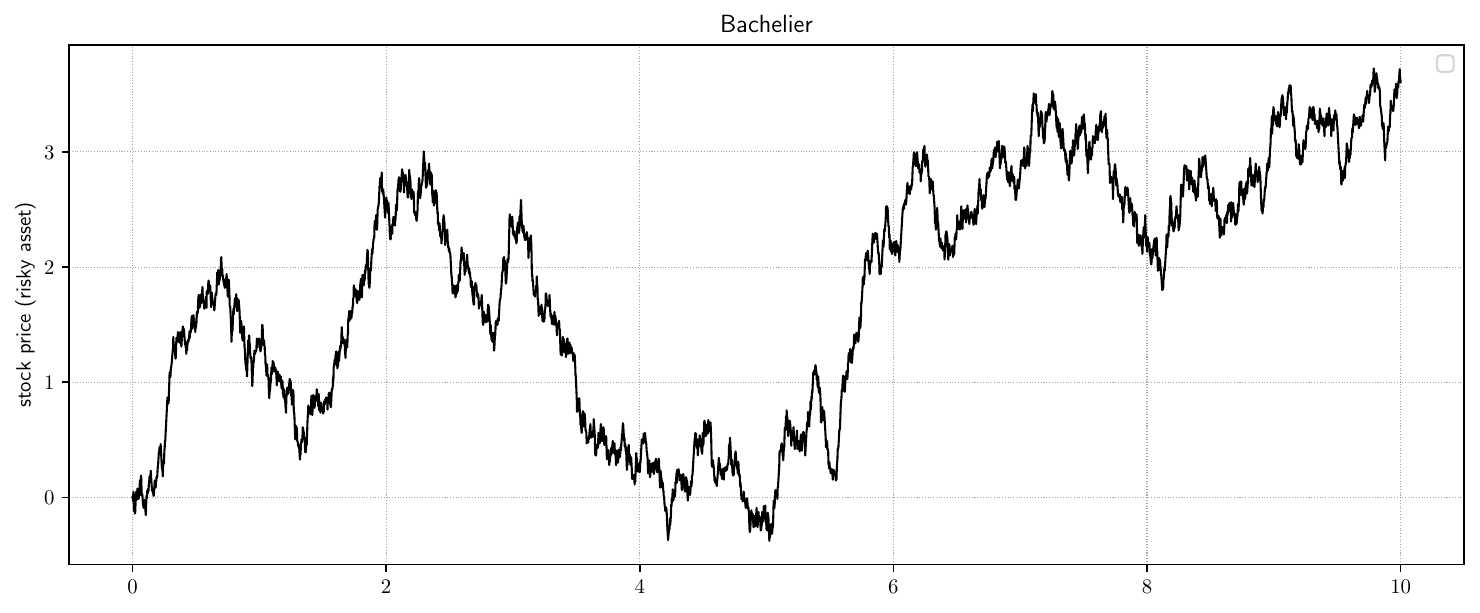}~\includegraphics[width=0.5\textwidth]{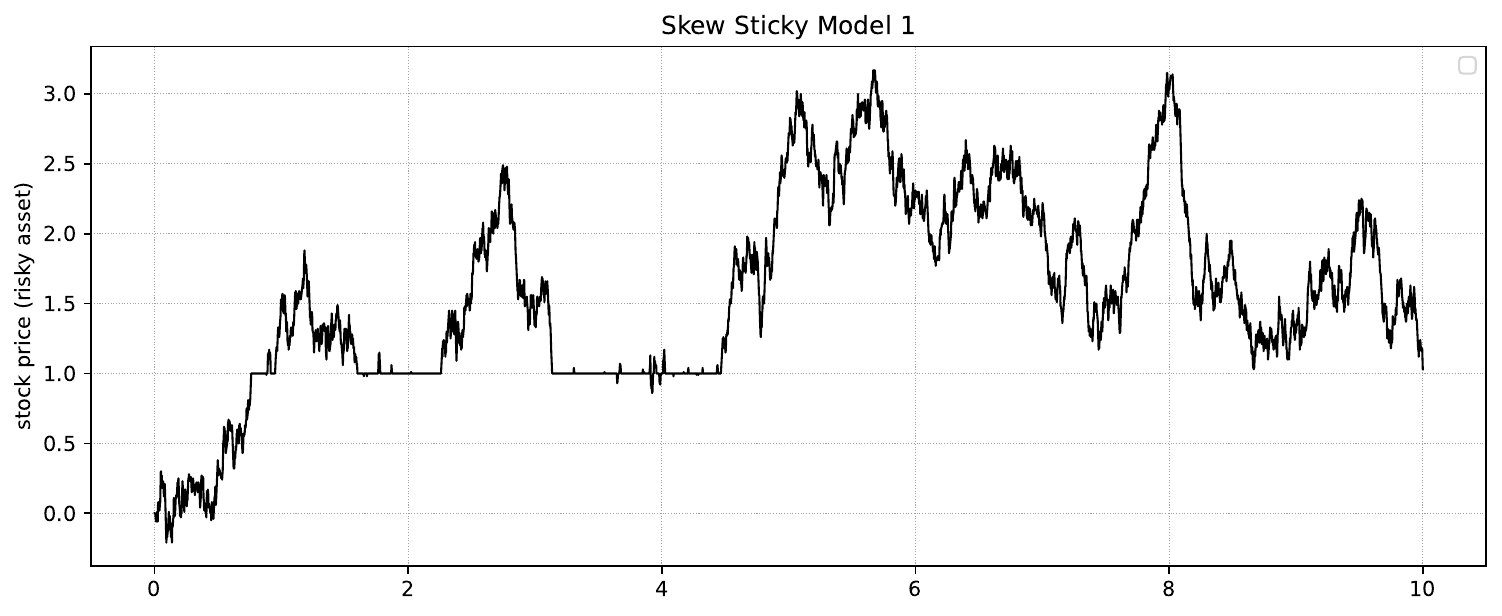}\\
		\includegraphics[width=0.5\textwidth]{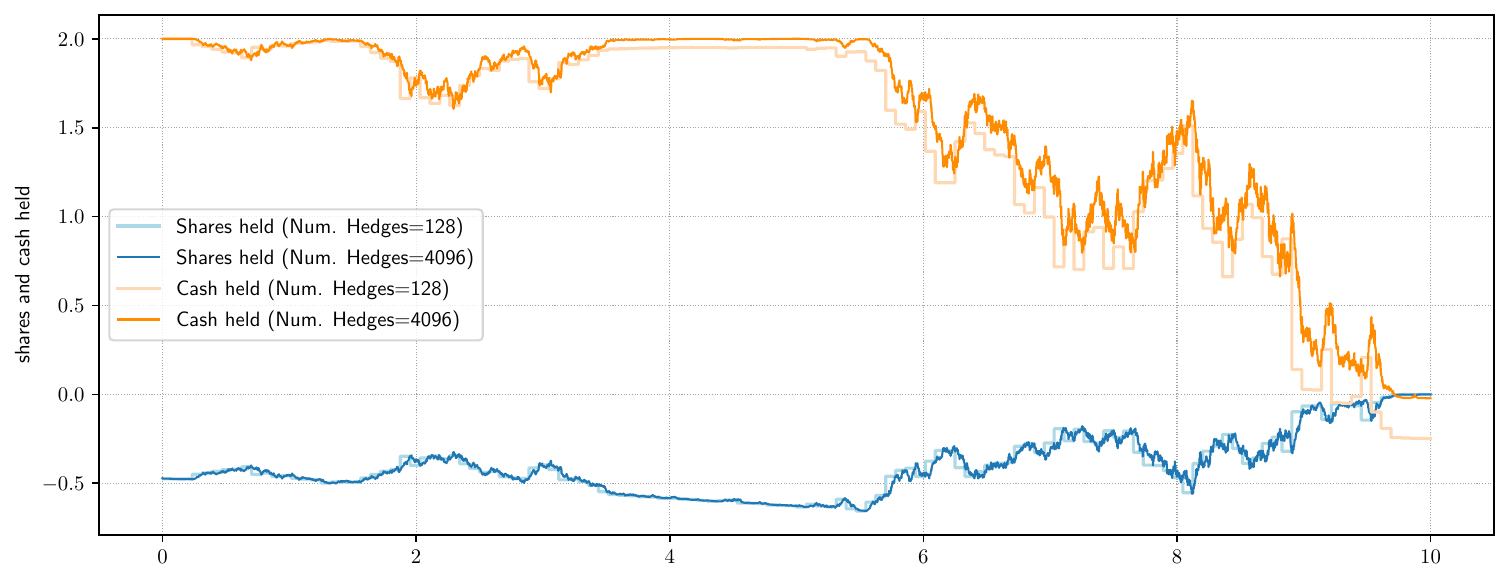}~\includegraphics[width=0.5\textwidth]{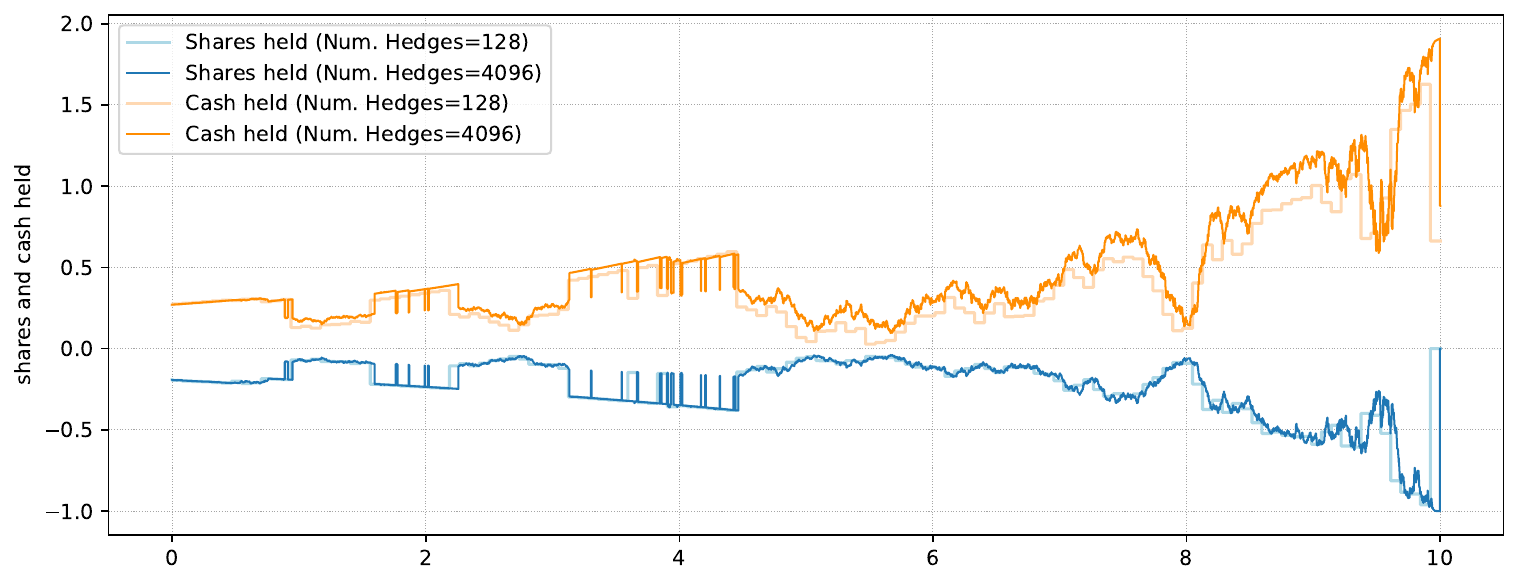}\\
		\includegraphics[width=0.5\textwidth]{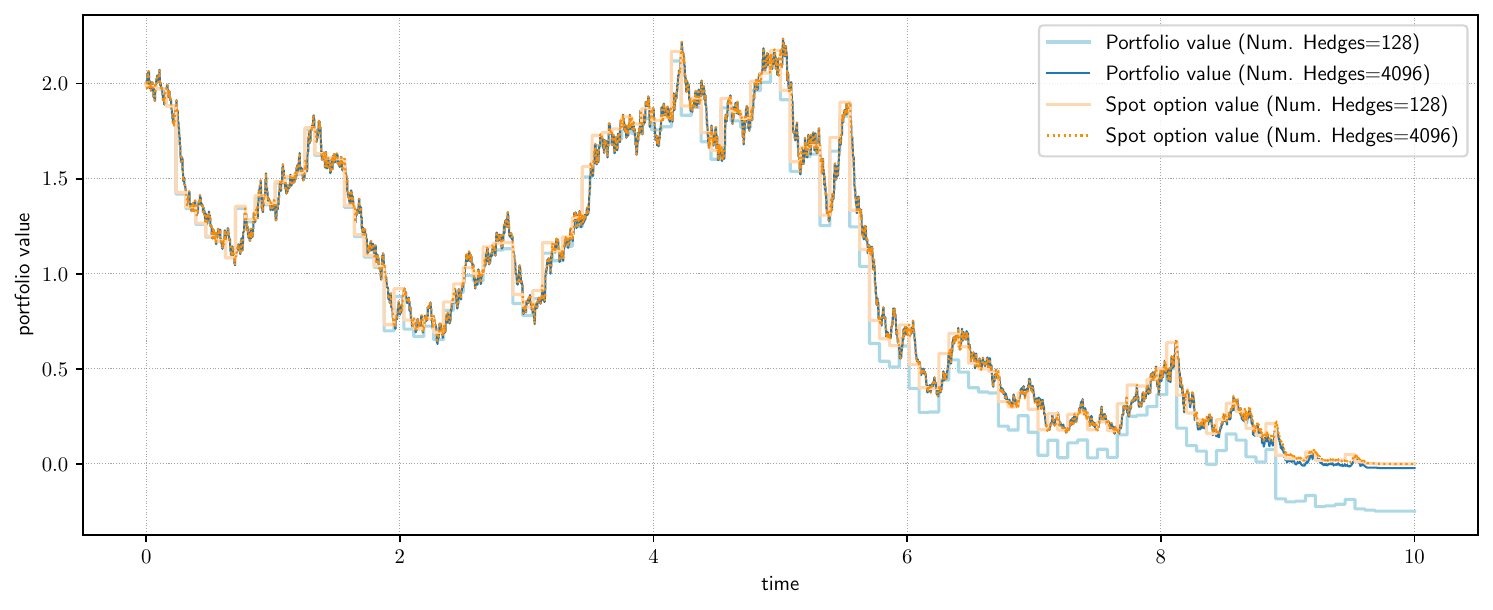}~\includegraphics[width=0.5\textwidth]{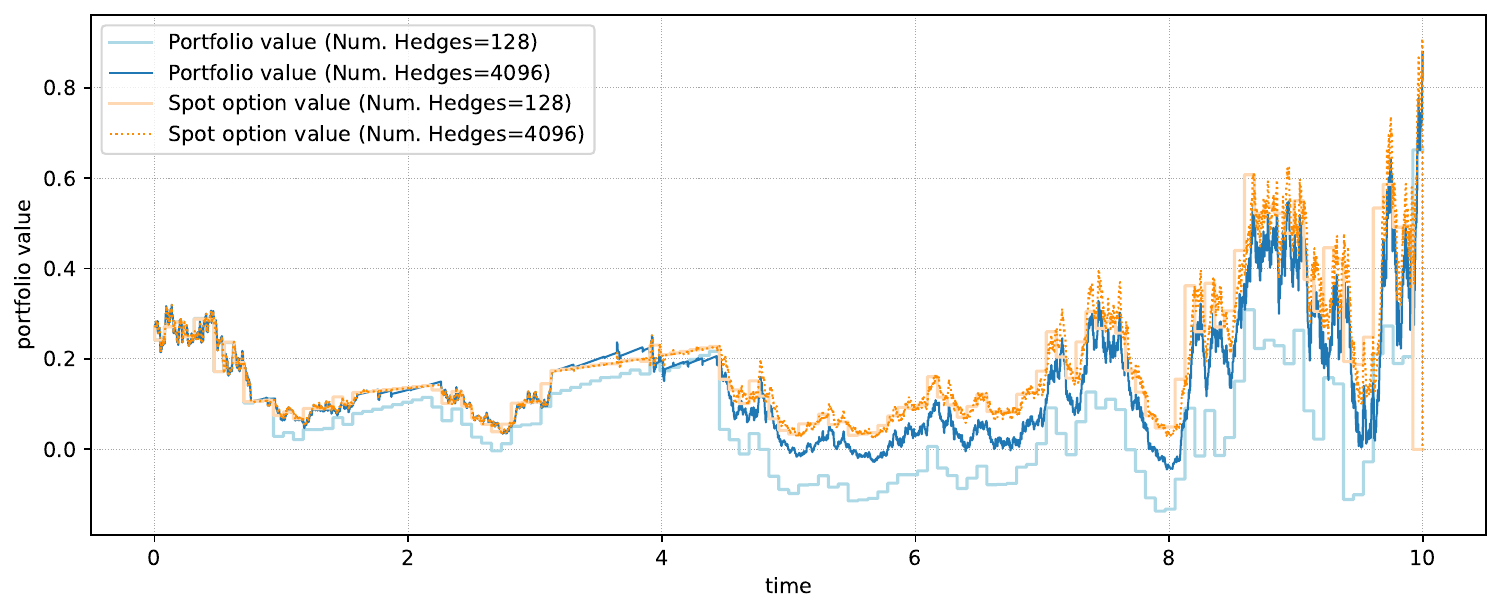}
		\caption{Single hedging scenario in the Bachelier model (left) and Skew-Sticky Model 1 (right). Rows display: (1st) price path of the risky asset, (2nd) shares of risky asset and cash holdings, (3rd) portfolio value (with premium subtracted).}
		\label{fig:single_hedge}
	\end{figure}
	
\begin{table}
	[htbp]
	\centering
	\caption{Tracking error of the single-scenario experiments depicted in Figure~\ref{fig:single_hedge}.}
	\label{table:tracking_error}
	\begin{tabular}{|c|c|c|}
		\hline
		\textbf{Model} & \textbf{Hedging instances} & \textbf{Tracking error} \\
		\hline
		Bachelier 			& 256	& -0.25 \\
		Bachelier 			& 4096  & -0.02 \\
		Skew Sticky Model~1 & 256 	& -0.31 \\
		Skew Sticky Model~2 & 4096  & -0.09 \\
		\hline
	\end{tabular}
\end{table}
	
	\begin{figure}[htbp]
		\includegraphics[width=0.35\textwidth]{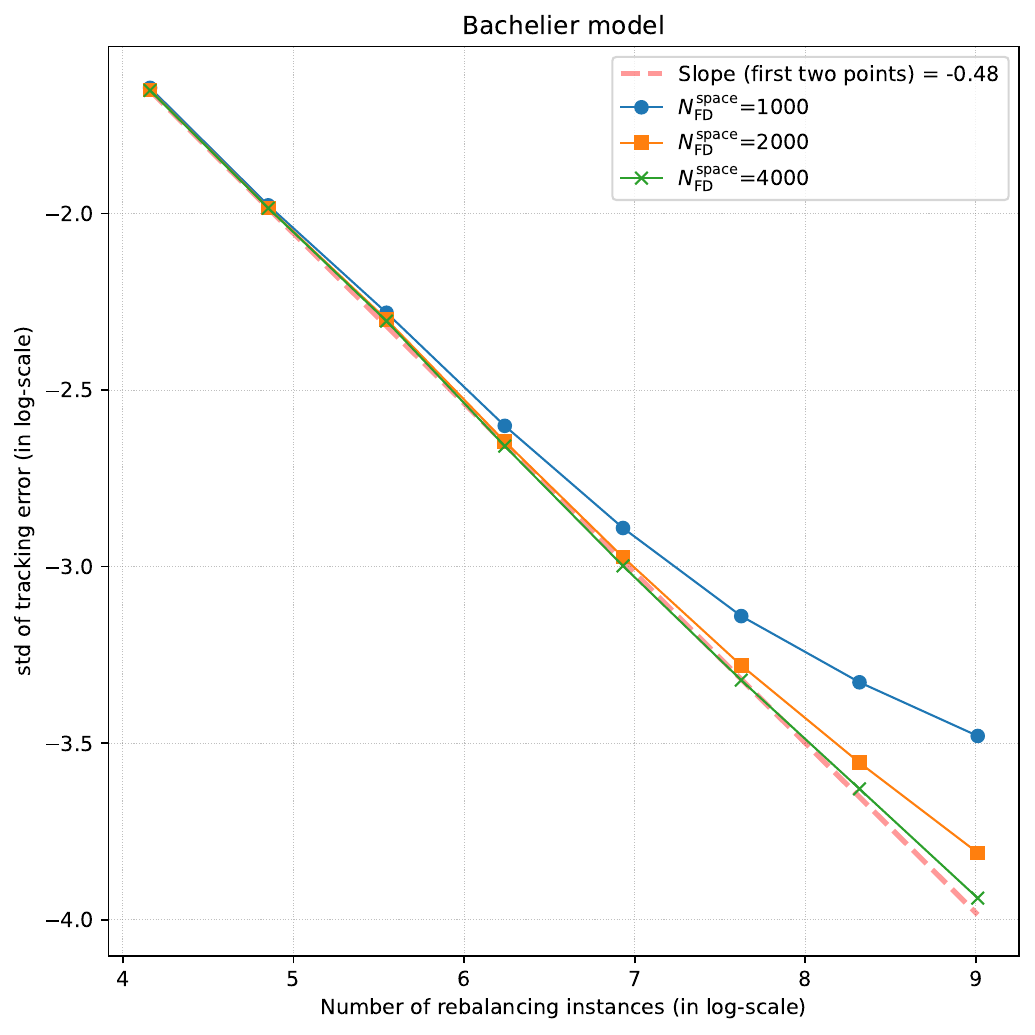}
		~\includegraphics[width=0.35\textwidth]{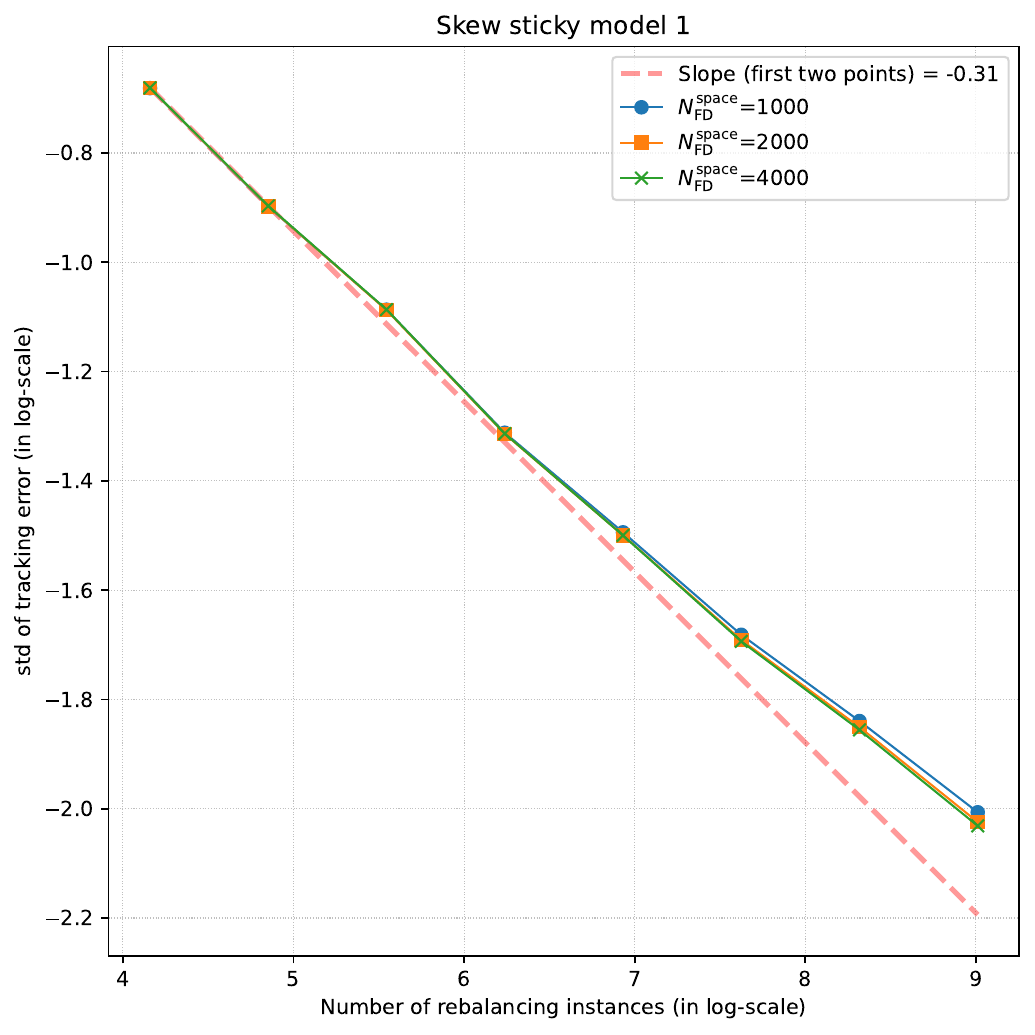}\\
		\includegraphics[width=0.35\textwidth]{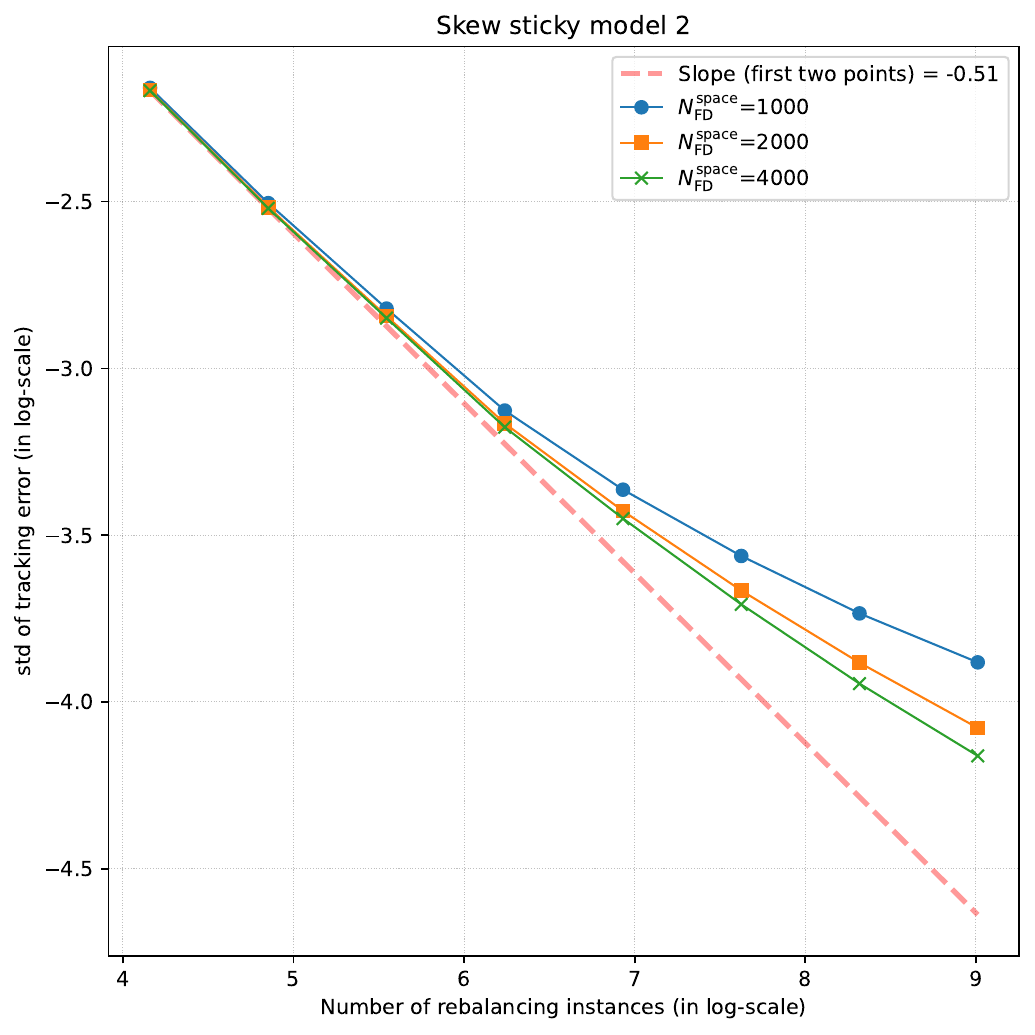}~\includegraphics[width=0.35\textwidth]{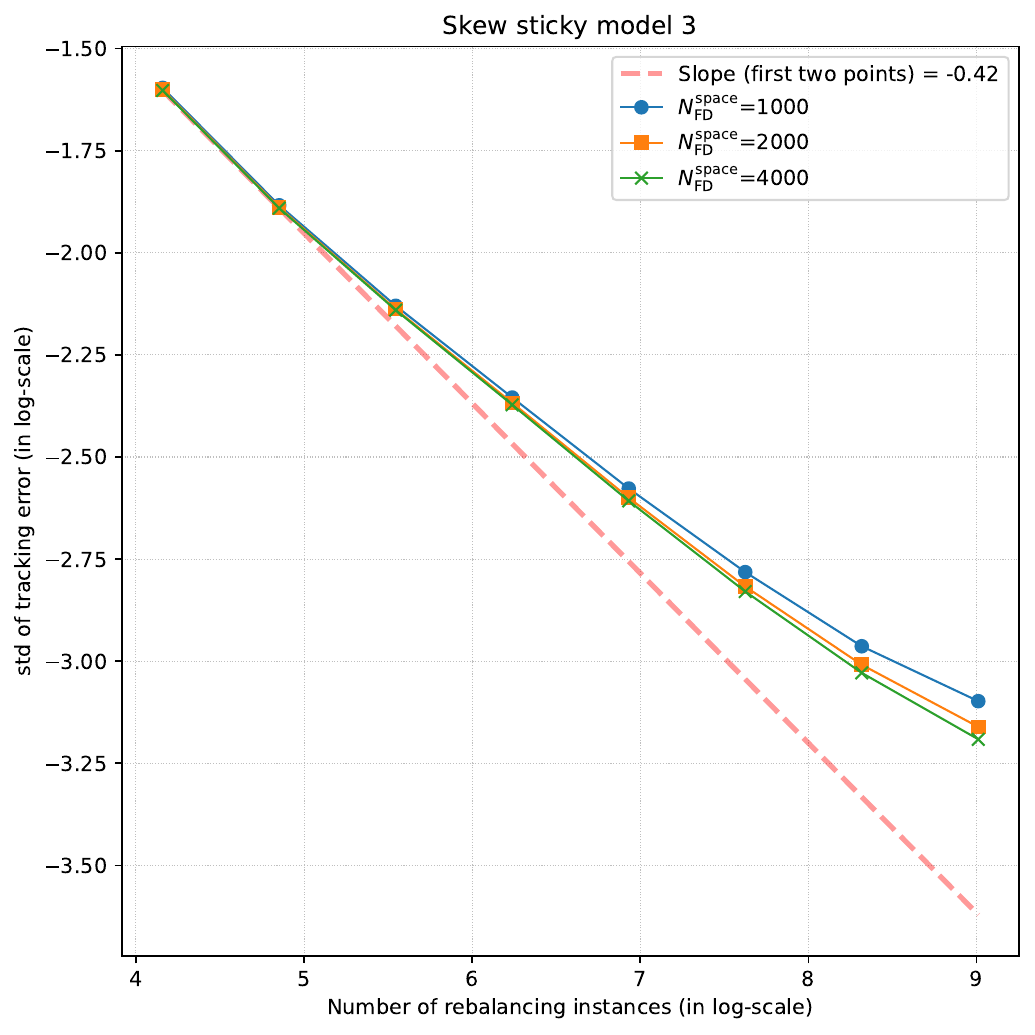}
		
		\caption{Log-log plot of the empirical standard deviation of the tracking error versus the number of hedging instances. 
		We consider various refinements of the space discretization in the finite difference approximation of the hedging PDE ($N^{\mathrm{space}}_{\FD} \in \{1000,2000,4000\} $ points).
		The time discretization in the finite difference scheme is fixed at $N^{\mathrm{time}}_{\FD} = 500$, further refinements do not change the approximated price and delta fields.
		The sample size of the Monte Carlo estimation if $N_{\mathrm{MC}} = 2000 $.
		}
		\label{fig:convergence}
	\end{figure}
	
	\begin{figure}[htbp]
		\includegraphics[width=0.85\textwidth]{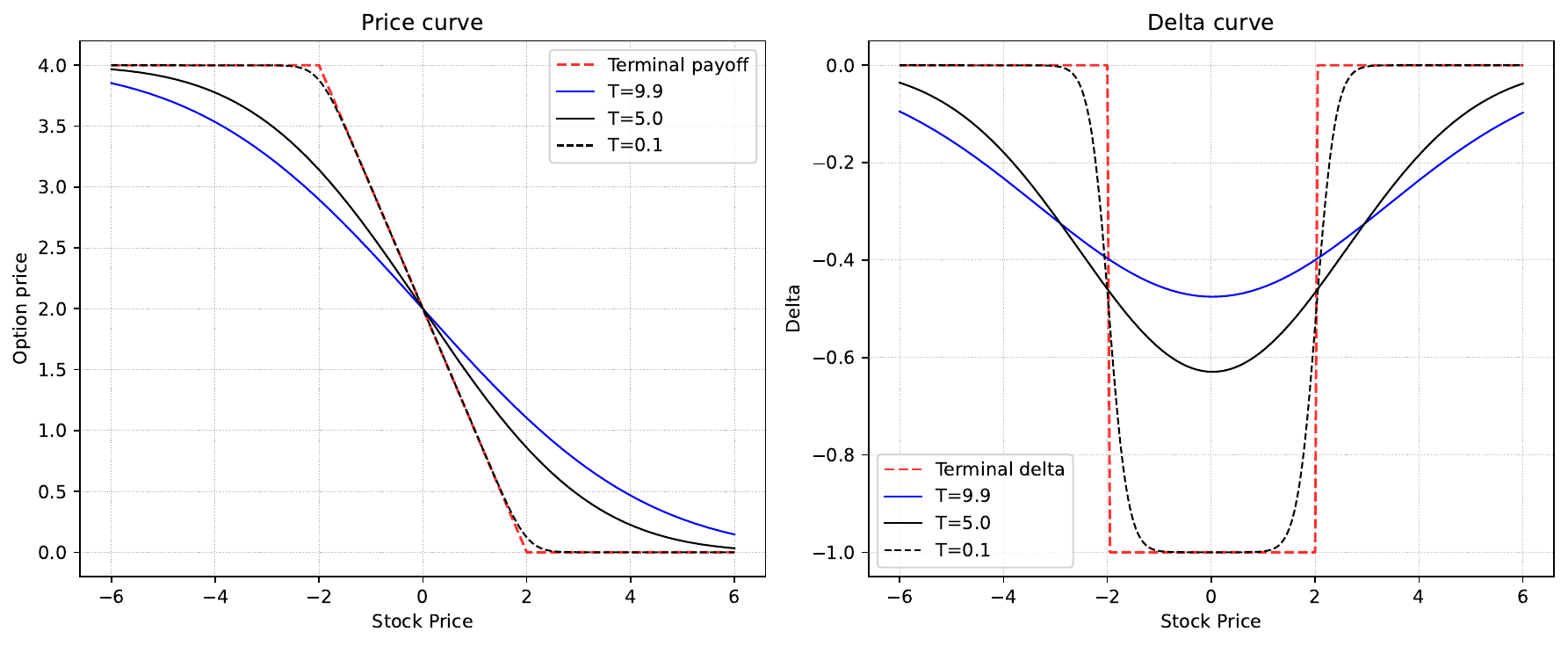}
		\\
		\includegraphics[width=0.85\textwidth]{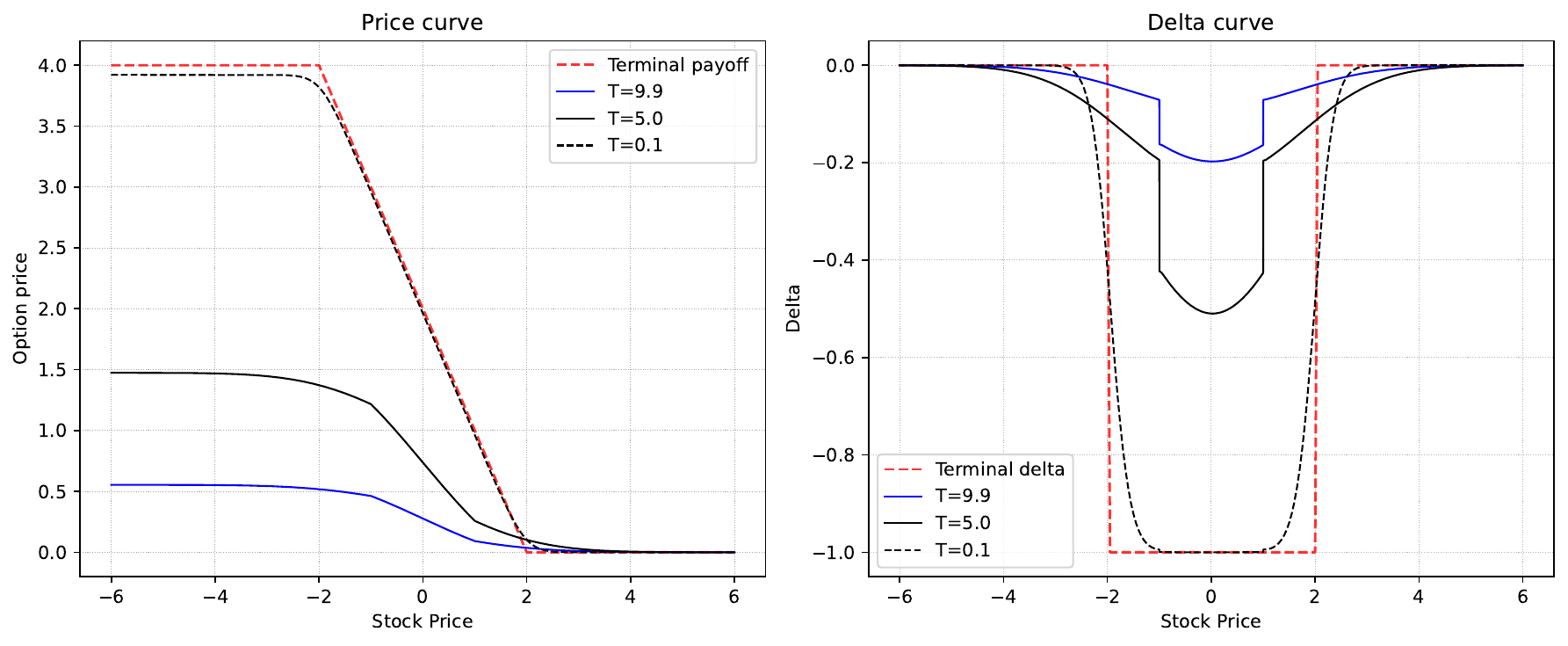}\\ 
		\includegraphics[width=0.85\textwidth]{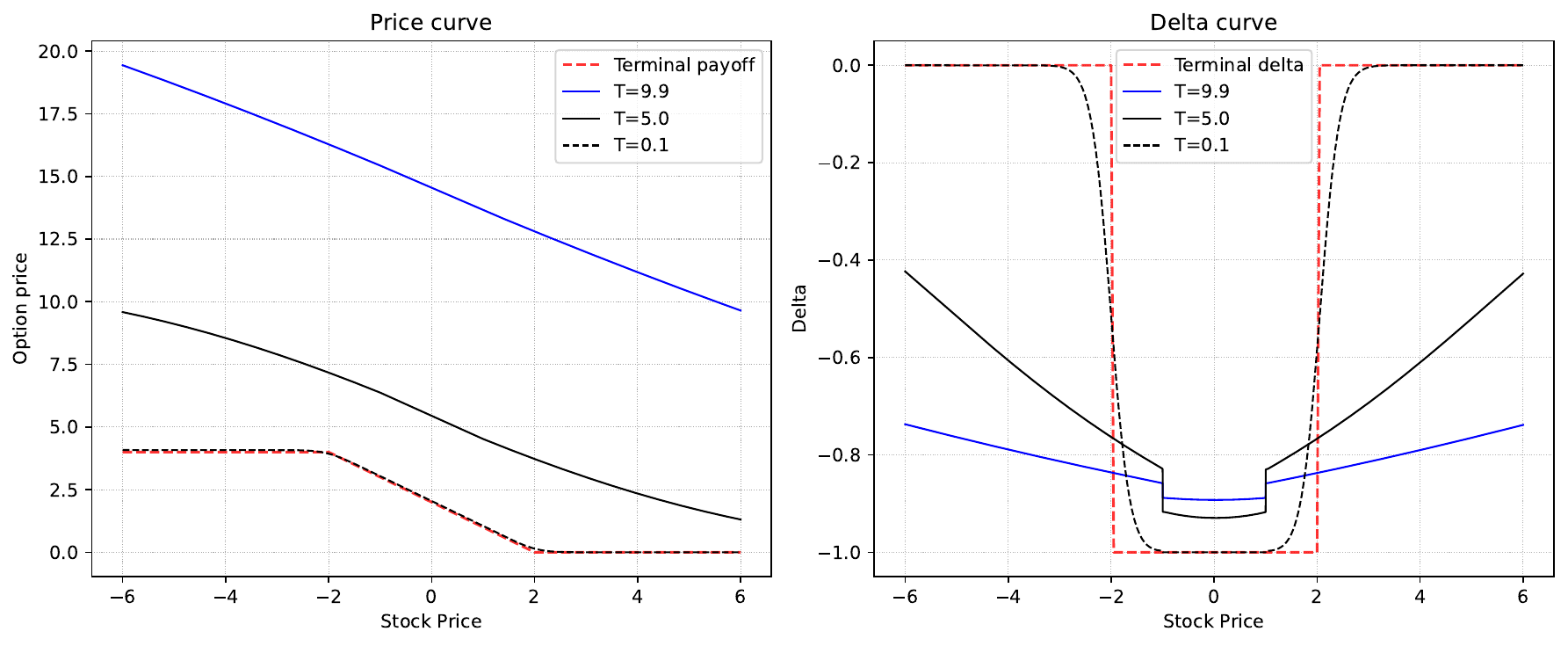}\\
		\includegraphics[width=0.85\textwidth]{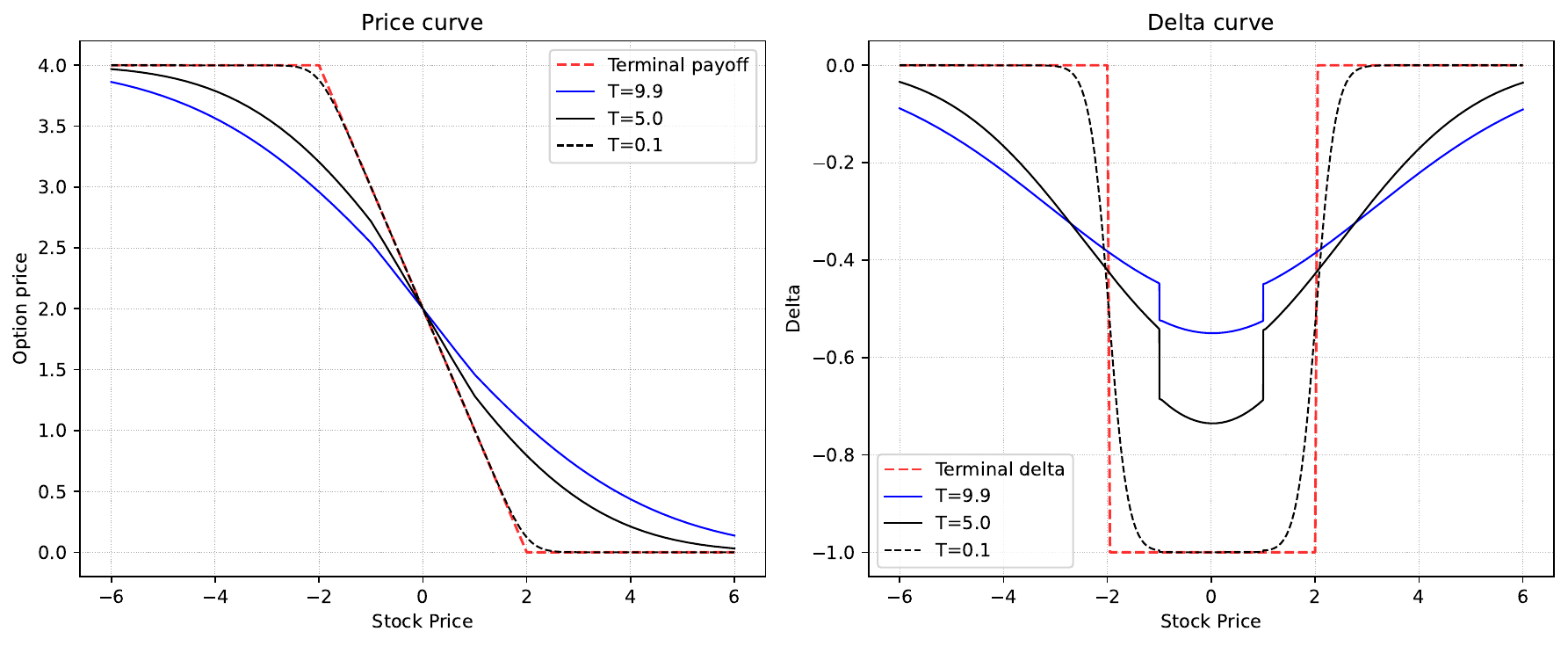}
		\caption{Finite difference approximations of price and delta curves for time horizons $T \in \{0.1, 5, 9.9\}$ in: (1st row) the Bachelier model, (2nd row) Skew-Sticky Model 1, (3rd row) Skew-Sticky Model 2, and (4th row) Skew-Sticky Model 3.}
		\label{fig:curves}
	\end{figure}
	
	\subsubsection{IFD approximation}
	
	The value and delta curves displayed in Figure~\ref{fig:curves} exhibit the expected qualitative behavior for a bear spread option across all four models. As maturity approaches, the price curve converges to the terminal payoff profile. The delta transitions smoothly from near zero in the out-of-the-money region to approximately $-1$ in the in-the-money region.
	
	At the skew-sticky points $x = \pm 1$, the delta exhibits characteristic jumps of~\eqref{eq: skew_sticky:hedging_equation}.
	
	\subsubsection{Performance of the Hedging Strategy}
	
	Table~\ref{tab:results} reports the Monte Carlo estimates of MTE* and StDTE* for all models and rebalancing frequencies. For all models, the mean tracking error remains close to zero (typically \(|\mathrm{MTE}^*|<0.07\)), confirming that the discrete self-financing strategy derived from the solution to~\eqref{eq: hedging PDE} is  replicates the option payoff on average. The standard deviation StDTE* decreases monotonically with \(N\), reflecting improved risk control through more frequent rebalancing. These observations are fully consistent with Theorem~\ref{theo: main1}.

	The Bachelier model serves as a smooth benchmark; its StDTE* decays roughly as \(N^{-1/2}\), as expected for classical diffusion models. The skew-sticky models exhibit larger StDTE* values, particularly for outward‐pointing skew (Model 1), indicating that the singular features of the diffusion introduce additional hedging errors. Model 2 (inward‐pointing skew) shows StDTE* values even lower than the Bachelier model at many frequencies, likely because the negative interest rate \(r=-0.2\) reduces the effective risk‐neutral volatility. 
	These qualitative differences highlight the influence of the scale and speed characteristics on the hedging error, as well as that of the interest rate (see Table~\ref{tab:ss1}).
	
	Figure~\ref{fig:single_hedge} illustrates the path-properties of the portfolio composition and value depending on the price of the risky asset in the Bachelier model and Skew-Sticky Model~1 for a single scenario with 4096 and 256 hedging instances (i.e.,
	10/4096 and 10/256 hedging windows).
	Tracking errors of these scenarios are shown in Table~\ref{table:tracking_error}. 
	
	\subsubsection{Robustness in the Absence of NFLVR}
	
	Skew-Sticky Model 3 violates the ELMM condition, yet the hedging algorithm achieves near‐zero MTE* and StDTE* values comparable to the Bachelier model for the same hedging frequencies. This provides empirical support for Theorem~\ref{theo: main1}, which states that the hedging strategy~\eqref{eq: main hedging strategy} is consistent irrespective of the existence of an ELMM. Although the initial capital may exceed the minimal hedging capital (as discussed in Section~\ref{ssec: skew_sticky}), the strategy still replicates the payoff with small tracking errors.
	
	\subsubsection{Convergence Trends}
	
	A log‐log plot of StDTE* versus \(N\) (Figure \ref{fig:convergence}) indicates that the decay is approximately \(N^{-1/2}\) for the Bachelier model, as expected from classical theory. For the skew-sticky models, the decay appears slower, especially for Model 1, which we attribute to the non‐smooth behaviour at the thresholds. A precise asymptotic analysis is beyond the scope of this illustrative study; we only note that the observed trends are qualitatively consistent with the theoretical framework.
	
	\subsection{Summary of Numerical Findings}
	
	The numerical experiments illustrate the main theoretical results. The discrete hedging strategy derived from the solution to the hedging PDE~\eqref{eq: hedging PDE} yields near‐zero mean tracking errors, and the standard deviation decreases with more frequent rebalancing. The strategy remains effective even when NFLVR fails, supporting the thesis of Theorem~\ref{theo: main1}. The qualitative differences across models reflect the influence of the scale and speed characteristics on the hedging error. These findings provide empirical confidence in the theoretical framework.
	
	\begin{algorithm}[htbp]
		\small
		\caption{Delta Hedging of an Option}\label{algo_deltahedging}
		\KwData{Stock price time-series: $(S_t)_{t \geq 0}$, Interest rate: $r$, Value field: $(t,x)\mapsto \texttt{value\_field}(t,x)$, Delta field: $(t,x)\mapsto \texttt{delta\_field}(t,x)$, Payoff function: $x\mapsto h(x)$, Time to maturity: $T$, Hedging window: $\Delta t$}
		\KwResult{Replication portfolio Profit and Losses (PnL), Tracking error}
		\medskip
		\tcp{INITIALIZATION}
		$t \leftarrow 0$\;
		$V_0 \leftarrow \texttt{value\_field}(S_0, T)$\;
		$\delta \leftarrow \texttt{delta\_field}(S_0, T)$\;
		$cash \leftarrow V_0 - (\delta \times S_0)$\;
		$shares\_held \leftarrow \delta$\;
		\medskip
		\tcp{MAIN LOOP}
		\While{$t < T$}{
			$t \leftarrow t + \Delta t$\;
			Retrieve current stock price $S_t$\;
			$cash \leftarrow cash \times e^{r \Delta t}$; \tcp*[f]{Accrue interest}\;
			$\delta \leftarrow \texttt{delta\_field}(S_t, T-t)$ \;
			$\Delta Shares \leftarrow \delta - shares\_held$ \;
			$cash \leftarrow cash - (\Delta Shares \times S_t)$; \tcp*[f]{Finance the share trade}\;
			$shares\_held \leftarrow \delta$; \tcp*[f]{Update share count}\;  
			$portfolio\_value \leftarrow cash + (shares\_held \times S_t)$; \tcp*[f]{Market value}\;
			Append $\delta$ to the list of deltas\;
			Append $portfolio\_value$ to the list of portfolio values\;
		}
		\medskip
		\tcp{FINAL STEP AT MATURITY ($t = T$)}
		$cash \leftarrow cash + (shares\_held \times S_T)$; \tcp*[f]{Liquidate stock position}\;
		\medskip
		\Return{PnL: $cash - \text{\upshape\texttt{value\_field}}(S_0,T)$, Tracking error: $h(S_T) - cash$}\;
	\end{algorithm}
	
	\section{Main proofs}
	\label{sec: proofs}
	
	\subsection{Proof of Theorem~\ref{theo: main1}}
	
	{\em Step 1:} 
	Let us first consider the function 
	\[
	\rr (t, x) := e^{r (T - t)} u (t, \g^{-1} (x)) = E^{Q_{\g^{-1} (x)}} \big[ h (\X_{T - t}) \big] = E^{P_x} \big[ h (\g^{-1} (\X_{T - t})) \big] 
	\]
	(recall~\eqref{eq:220826a3}).
	Because \(h \in L^\infty (\om)\), we have \(h \circ \g^{-1} \in L^\infty (\tm)\). The function \(\rr\) is therefore ``structurally comparable'' to \(u\), but related to a diffusion on natural scale (namely, \((x \mapsto P_x)\)). Hence, we can deduce properties of \(\rr\) from fundamental results by H.~P.~McKean \cite{MK_56}. 
	To wit, thanks to \cite[Corollaries~4.3 and 4.4]{MK_56}, \(u\) is continuous on \([0, T) \times G(J)\), the time-derivative \(\rr_t\) exists as a continuous function on \([0, T) \times \g(J)\), the left-derivative \(\partial^-_x \rr\) exists and 
	\begin{align} \label{eq: PDE McKean r}
		\partial^-_x \rr (t, x) - \partial^-_x \rr (t, y) = - \int_{[y, x)} 2\rr_t (t, z) \, \tm (\rd z), \quad t \in [0, T), \, x, y \in \g (J), \, y < x.
	\end{align} 
	Of course, using that \(\g\) is continuous, it follows directly that 
	\[
	u_t (t, x) = e^{- r (T - t)} \rr_t (t, \g (x)) + r e^{- r (T - t)} \rr (t, \g (x))
	\] 
	exists as a continuous function on \([0, T) \times J\), and 
	\begin{align} \label{eq: Z1}
		\int_{[ \g (y), \g (x))} \rr_t (t, z) \, \tm (\rd z) = \int_{[y, x)} e^{r (T - t)} \big( u_t (t, z) - r u (t, z) \big) \, \om (\rd z), \quad x, y \in J, \, y < x.
	\end{align} 
	Further, as \(\g\) is strictly increasing with left-hand derivative~\(\g'_-\), the left-hand derivative \(\partial^-_x u\) exists with
	\begin{align*}
		\partial^-_x u (t, x) 
		&= \lim_{h \searrow 0} e^{- r (T - t)} \frac{ \rr (t, \g (x)) - \rr (t, \g (x - h))}{\g (x) - \g (x - h)} \frac{\g (x) - \g (x - h)}{h} 
		\\&= e^{- r (T - t)} \partial^-_x \rr (t, \g(x)) \g'_- (x). 
	\end{align*} 
	Putting these pieces together, for all \(x, y \in J\) with \(y < x\), we obtain that 
	\begin{align*} 
		\frac{\partial^- u}{\partial G} (t, x) - \frac{\partial^- u}{\partial G} (t, y) &= \frac{\partial^-_x u (t, x)}{\g'_- (x)} - \frac{\partial^-_x u (t, y)}{\g'_- (y)} 
		\\&= e^{- r (T - t)} \big( \partial^-_x \rr (t, \g (x)) - \partial^-_x \rr (t, \g (y)) \big)
		\\&= - \int_{[\g (y), \g (x))} 2 e^{- r (T - t)} \rr_t (t, z) \, \tm (\rd z)
		\\&= - \int_{[ y, x)} 2 \big( u_t (t, z) - ru (t, z) \big) \, \om (\rd z).
	\end{align*} 
	To conclude that \(u\) is a good solution to \eqref{eq: backward PDE}, it remains to prove that \(\partial^-_x u\) is locally bounded on \([0, T) \times J\), continuous in the first and LCRL in the second variable (when the others remain fixed). By the continuity and monotonicity of \(\g\) and the fact that \(x \mapsto \g'_- (x)\) is LCRL, the formula \(\partial^-_x u (t, x) = e^{- r (T - t)} \partial^-_x \rr (t, \g (x)) \g'_-(x)\) shows that it suffices to show that \(\partial^-_x \rr\) is locally bounded on \([0, T) \times G (J)\), continuous in the first and LCRL in the second variable. This is the program for the remainder of this step. First, the LCRL property in the second variable follows from \eqref{eq: PDE McKean r}, the fact that \(\tm\) is locally finite on \(G(J)\), and the dominated convergence theorem. We proceed with the continuity in the first variable. 
	We take an arbitrary \(x_0 \in \g(J)\) such that \(\tm (\{x_0\}) = 0\). Of course, such a point must exist. Furthermore, we take two points \(a, b \in \g(J)\) such that \(a < x_0 < b\) and set \(T_{a, b} := \inf \{ t \geq 0 \colon \X_t = a \text{ or } \X_t = b \}\). Finally, take a time point \(t_0 \in [0, T)\). 
	We notice that \(x \mapsto \rr (t_0, x)\) is in the domain of the (extended) generator of \((x \mapsto P_x)\), which follows from results in \cite[Section~2.7, p. 131]{freedman}. Consequently, by Dynkin's formula (\cite[Lemma~48, p. 119]{freedman}) and standard properties of the speed measure (\cite[Lemma~68, p. 128]{freedman}), 
	\[
	E^{P_{x_0}}\big[ \rr (t_0, \X_{T_{a, b}}) \big] = \rr (t_0, x_0) - \int_a^b G_{a, b} (x_0, y) \rr_t (t_0, y) \, \tm(\rd y), 
	\]
	where \(G_{a, b}\) is the Green function given by 
	\[
	G_{a, b} (x, y) := \frac{ 2 (x \wedge y - a) (b - x \vee y)}{b - a}, \quad x, y \in [a, b]. 
	\] 
	With
	\[
	E^{P_{x_0}} \big[ \rr (t_0, \X_{T_{a, b}}) \big] = \rr (t_0, b) \frac{x_0 - a}{b - a} + \rr (t_0, a) \frac{b - x_0}{b - a}, 
	\]
	we conclude that 
	\begin{align*}
		\rr (t_0, b) (x_0 - a) + \rr (t_0, a) (b - x_0) = (b - a) \Big( \rr (t_0, x_0) - \int_a^b G_{a, b} (x_0, y) \rr_t (t_0, y) \, \tm(\rd y) \Big). 
	\end{align*}
	For \(\varkappa > 0\) such that \(a < x_0 - \varkappa\), we get in the same way that 
	\begin{align*}
		\rr (t_0, b) (x_0 - \varkappa - a) &+ \rr (t_0, a) (b - x_0 + \varkappa) 
		\\&= (b - a) \Big( \rr (t_0, x_0 - \varkappa) - \int_a^b G_{a, b} (x_0 - \varkappa, y) \rr_t (t_0, y) \, \tm(\rd y) \Big). 
	\end{align*}
	Subtracting these equations, we get that 
	\begin{align*} 
		\rr (t_0, b) - \rr (t_0, a) = (b - a) \Big( &\frac{ \rr (t_0, x_0) - \rr (t_0, x_0 - \varkappa) }{\varkappa} \\&+ \int_a^b \frac{G_{a, b} (x_0 - \varkappa, y) - G_{a, b} (x_0, y) }{\varkappa} \, \rr_t (t_0, y) \, \tm(\rd y) \Big), 
	\end{align*}  
	and taking \(\varkappa \to 0\) we arrive at 
	\begin{align*}
		\partial^-_x \rr (t_0, x_0) = \frac{\rr(t_0, b) - \rr (t_0, a)}{b - a} &- 2 \int_{[a, x_0)} \frac{(y - a)}{(b - a)} \, \rr_t (t_0, y) \,\tm (\rd y) 
		\\&+ 2 \int_{[x_0, b)} \frac{ (b - y)}{(b - a) } \, \rr_t (t_0, y) \, \tm (\rd y). 
	\end{align*}
	We learned the strategy of getting this formula from the proof of \cite[Theorem~16.64]{breiman1968probability}.
	Recalling that \(\rr_t\) is continuous on \([0, T) \times G (J)\), this equation proves that \(t \mapsto \partial^-_x \rr (t, x_0)\) is continuous on~\([0, T)\).
	Using \eqref{eq: PDE McKean r}, for any \(t \in [0, T)\) and \(x \in \g(J)\), we find that 
	\begin{align*} 
		\partial^-_x \rr (t, x \vee x_0) & =  \partial^-_x \rr (t, x \vee x_0) - \partial^-_x \rr (t, x \wedge x_0) + \partial^-_x \rr (t, x \wedge x_0) 
		\\&= - \int_{[ x \wedge x_0, x \vee x_0)} 2 \rr_t (t, z) \, \tm (\rd z) + \partial^-_x \rr (t, x \wedge x_0).
	\end{align*} 
	By dominated convergence, and using the fact that \(t \mapsto \partial^-_x \rr (t, x_0)\) is continuous, this implies that \(t \mapsto \partial^-_x \rr (t, x)\) is continuous on \([0, T)\) for all \(x \in G (J)\). Finally, we prove the local boundedness on \([0, T) \times G (J)\). Take \(t \in [0, T)\) and \([a, b] \subset G (J)\). Then, again by \eqref{eq: PDE McKean r}, for every \(s \in [0, t]\) and \(x \in [a, b]\), we obtain that 
	\begin{align*} 
		| \partial^-_x \rr  (s, x) | &= \Big| \partial^-_x \rr (s, a) - \int_{[a, x)} 2 \rr_t (s, z) \, \tm (\rd z) \Big| 
		\\&\leq \sup_{h \in [0, t]}| \partial^-_x \rr (h, a) | + 2 \sup_{\substack{h \in [0, t] \\ z \in [a, b]}} | \rr_t (h, z) | \, \tm ([a, b]) < \infty, 
	\end{align*} 
	where finiteness follows from the fact that \(s \mapsto \partial^-_x \rr (s, a)\) is continuous on \([0, T)\), \(\rr_t\) is continuous on \([0, T) \times G (J)\), and \(\tm\) is locally finite on \(G (J)\). 
	In summary, part (a) is proved.
	
	\smallskip
	{\em Step 2:}
	Next, we establish part (b). 
	Let \(v\) be a good solution in the sense of Definition~\ref{def: good solution}. Then, for every \(\varepsilon \in (0, T)\), the mapping \([0, T - \varepsilon] \times G(J) \ni (t, x) \mapsto \overline{v} (t, x) := e^{r (T - t)} v (t, G^{-1}(x))\) satisfies the prerequisites of the It\^o-type formula \cite[Theorem~4.1, Definition~3.3]{Wil_18}, which can be viewed as a semimartingale version of the Az{\'e}ma--Jeulin--Knight--Yor~\cite{AJKY_98} extension of the generalized It\^o formula to space-time functions of Brownian motion. 
	At this point, we remark that \cite[Theorem~4.1]{Wil_18} is formulated for \(\bR\)-valued semimartingales and functions defined on \(\bR_+ \times \bR\), but the usual localization argument shows that we can also apply it to \(\g (J)\)-valued semimartingales and functions defined on \([0, T - \varepsilon] \times \g (J)\). 
	Take \(x_0 \in J\) and let \(\varepsilon \in (0, T)\). From part (d) of Definition~\ref{def: good solution}, for all \(t \in [0, T)\) and \(x, y \in G(J)\) with \(x > y\), we get that 
	\begin{equation} \label{eq: second derivative measure in proof}
		\begin{split}
			\partial^-_x \overline{v} (t, x) - \partial^-_x \overline{v} (t, y) &= e^{r (T - t)} \Big( \, \frac{\partial^- v}{\partial G} (t, G^{-1} (x)) - \frac{\partial^- v}{\partial G} (t, G^{-1} (y)) \Big) \phantom \int
			\\&= e^{r (T - t)} \int_{[G^{-1}(y), G^{-1} (x))} 2 \big( r v (t, z) - v_t (t, z)\big) \, \om (\rd z)
			\\&= - \int_{[y,x)} 2 \overline{v}_t (t, z) \, \tm (\rd z).
		\end{split}
	\end{equation} 
	Running \(t\) through \([0, T - \varepsilon]\), using \cite[Theorem~4.1, Definition~3.3]{Wil_18} with \eqref{eq: second derivative measure in proof}, \cite[Exercise~VI.1.23]{RY} and \cite[Theorem~V.49.1]{RW2}, we obtain that
	\begin{alignat}{2} 
		\rd e^{r (T - t)} v (t, \Y_t) &= \overline{v}_t (t, \g(\Y_t)) \vd t &&+\partial^-_x \overline{v} (t, \g(\Y_t)) \vd  \g (\Y_t) - \frac{1}{2} \int_{\g(J)} 2\overline{v}_t (t, y) \vd L^y_t (\g(\Y)) \, \tm (\rd y)
		\\&=  \overline{v}_t (t, \g(\Y_t)) \vd t &&+ \partial^-_x \overline{v} (t, \g(\Y_t)) \vd  \g (\Y_t) 
		\\&&&- \int_{J} \overline{v}_t (t, \g (x)) \vd L^{\g (x)}_t (\g (\Y)) \frac{\s'_+ (x)}{\G (x)} \, \m (\rd x) 
		\\&= \overline{v}_t (t, \g(\Y_t)) \vd t &&+ \partial^-_x \overline{v} (t, \g(\Y_t)) \vd  \g (\Y_t) \label{eq: important computation uniqueness}
		\\&&&- \int_{\s (J)} \overline{v}_t (t, \g (\s^{-1} (x))) \vd L^{x}_t (\s(\Y)) \, \m \circ \s^{-1} (\rd x) 
		\\&= \overline{v}_t (t, \g(\Y_t)) \vd t &&+ \partial^-_x \overline{v} (t, \g(\Y_t)) \vd  \g (\Y_t) - \overline{v}_t (t, \g (\Y_t)) \vd t \phantom \int 
		\\&= \partial^-_x \overline{v} (t, \g(\Y_t)) \vd  &&\g(\Y_t).  \phantom \int
	\end{alignat}
	Recall that \(\g\) is a dc function and 
	\[
	\g'_- (x) - \g'_- (y) = \int_{[y, x)} 2 \g'_- (z) \, \nu (\rd z), \quad y, x \in J, y < x.
	\] 
	Applying the generalized It\^o formula \cite[Theorem~VI.1.5]{RY}, the formula from \cite[Exercise~VI.1.23]{RY}, and using the occupation times formula for diffusions \cite[Theorem~V.49.1]{RW2}, we obtain that 
	\begin{align} 
		\rd \g (\Y_t) 
		&= \g'_- (\Y_t) \vd  \Y_t + \int_{J} \rd L^x_t (\Y) \g'_- (x) \, \hm (\rd x) 
		\\&= \g'_- (\Y_t) \vd  \Y_t - \int_{J} \rd L^x_t (\Y) \s'_+ (x) \, r x \g'_- (x) \, \m (\rd x) 
		\\&= \g'_- (\Y_t) \vd  \Y_t - \int_J \rd L^{\s (x)}_t (\s(\Y)) r x \g'_- (x) \, \m (\rd x) 
		\\&= \g'_- (\Y_t) \vd  \Y_t - \int_{\s (J)} \rd L^{x}_t (\s(\Y)) r \s^{-1} (x) \g'_- (\s^{-1} (x)) \, \m \circ \s^{-1} (\rd x) 
		\\&= \g'_- (\Y_t) \vd  \Y_t - r \s^{-1} (\s (\Y_t)) \g'_- (\s^{-1} (\s (\Y_t))) \vd  t \phantom \int
		\\&= \g'_- (\Y_t) \vd  \Y_t - r \Y_t \g'_- (\Y_t) \vd t.\phantom \int \label{eq: g (Y) dynamics} 
	\end{align}
	Combining \eqref{eq: g (Y) dynamics} with \eqref{eq: important computation uniqueness}, we obtain that, on \([0, T - \varepsilon]\), 
	\begin{align*} 
		\rd e^{r (T - t)} v (t, \Y_t) &= \partial^-_x \overline{v} (t, \g (\Y_t)) \g'_- (\Y_t) \vd \Y_t - \partial^-_x  \overline{v} (t, \g (\Y_t)) r \Y_t \g'_- (\Y_t) \vd t
		\\&= e^{r (T - t)} \Big( \partial^-_x v (t, \Y_t) \vd \Y_t - r \partial^-_x v (t, \Y_t) \Y_t \vd t \Big), 
	\end{align*}
	which yields that 
	\begin{align*}
		\rd v (t, \Y_t) &= \rd  e^{rt} e^{- rt} v (t, \Y_t) = e^{rt} \vd (e^{- rt} v (t, \Y_t) ) + e^{- rt} v (t, \Y_t) \vd \Y^0_t 
		\\&= e^{rt} e^{- rt} \Big( \partial^-_x v (t, \Y_t) \vd \Y_t - r \partial^-_x v (t, \Y_t) \Y_t \vd t \Big) + e^{- rt} v (t, \Y_t) \vd \Y^0_t 
		\\&= \partial^-_x v (t, \Y_t) \vd \Y_t - \partial^-_x v (t, \Y_t) \Y_t e^{- rt} \vd \Y^0_t + e^{- rt} v (t, \Y_t) \vd \Y^0_t 
		\\&= \partial^-_x v (t, \Y_t) \vd \Y_t + \frac{v (t, \Y_t) - \partial^-_x v (t, \Y_t) \Y_t}{\Y^0_t} \vd \Y^0_t.
	\end{align*} 
	As \(\varepsilon \in (0, T)\) is arbitrary, this gives the claimed formula \eqref{eq: 1st hedge}. 
	For the final claim, take \(h \in C_b (\J)\). As \(u\) is a good solution by part (a), the formula  \eqref{eq: 1st hedge} also holds with \(v\) replaced by \(u\).
	As \(h \in C_b (\J)\), the map
	\[
	\bR_+ \times J^* \ni (t, x) \mapsto E^{Q_x} \big[ h (\X_t) \big]
	\] 
	is continuous by \cite[Lemma~30, p.~116]{freedman}. Thus, sending \(t \to T\), we find that a.s. 
	\[
	u (t, \Y_t) = e^{- r (T - t)} E^{Q_{\Y_t}} \big[ h (\X_{T - t}) \big] \to E^{Q_{\Y_T}} \big[ h (\X_0) \big] = h (\Y_T), 
	\] 
	which implies the formula \eqref{eq: 2nd hedge}. The proof is complete. \qed

	\subsection{Proof of Theorem~\ref{theo: uniqueness pricing eq}}
	
	Let \(v\) be a bounded continuous good solution to \eqref{eq: backward PDE} with \(v (T, x) = h (x)\) for all \(x \in \I\), and take \(x_0 \in J, t_0 \in [0, T)\) and \(\varepsilon \in (0, T - t_0)\). 
	Then, as in~\eqref{eq: important computation uniqueness},\footnote{We use the assumption \(\J = \I\) to apply the It\^o formula from \cite[Theorem~4.1]{Wil_18}.}  running \(t\) through \([0, T - \varepsilon - t_0]\), we obtain \(Q_{x_0}\)-a.s.
	\begin{align*}
		\rd v &(t_0 + t, \X_t) 
		\\&= v_t (t_0 + t, \X_t) \vd t + \frac{\partial^-_x v (t_0 + t, \X_t)}{\g'_- (\X_t)} \vd \g(\X_t) 
		\\&\hspace{1cm} + \int_{\g (J)} \, (r v (t_0 + t, \g^{-1} (z)) - v_t (t_0 + t, \g^{-1} (z)) ) \vd L^z_t (\g(\X)) \, \om \circ \g^{-1} (\rd z)
		\\ &= v_t (t_0 + t, \X_t) \vd t + \frac{\partial^-_x v (t_0 + t, \X_t)}{\g'_- (\X_t)} \vd \g(\X_t) + (r v (t_0 + t, \X_t) - v_t (t_0 + t, \X_t) ) \vd t
		\\&= \frac{\partial^-_x v (t_0 + t, \X_t)}{\g'_- (\X_t)} \vd \g(\X_t) + rv (t_0 + t, \X_t) \vd t.
	\end{align*}
	Hence, integration by parts yields that \(Q_{x_0}\)-a.s. 
	\[
	\rd \big(e^{r (T - t_0 - t)} v (t_0 + t, \X_t)\big) = e^{r (T - t_0 - t)} \frac{\partial^-_x v (t_0 + t, \X_t)}{\g'_- (\X_t)} \vd \g(\X_t).
	\]
	As \(\g (\X)\) is a continuous local \(Q_x\)-martingale (by \cite[Corollary~V.46.15]{RW2}), the process \[(e^{r (T - t_0 - t)} v (t_0 + t, \X_t))_{t \leq T - \varepsilon - t_0}\] is a bounded local \(Q_{x_0}\)-martingale, hence a true \(Q_{x_0}\)-martingale. Consequently, 
	\[
	e^{r (T - t_0)} v (t_0, x_0) = E^{Q_{x_0}} \big[ e^{r \varepsilon} v (T - \varepsilon, \X_{T - \varepsilon - t_0}) \big].        
	\]
	By the assumed continuity of \(v\) (as function on \([0, T] \times J\)), letting \(\varepsilon \to 0\) and using the dominated convergence theorem, 
	\[
	e^{ r (T - t_0)} v (t_0, x_0) = E^{Q_{x_0}} \big[ v (T, \X_{T - t_0}) \big] = E^{Q_{x_0}} \big[ h (\X_{T - t_0}) \big]. 
	\]
	As \((t_0, x_0)\) was arbitrary, this proves that the PDE \eqref{eq: backward PDE} has at most one solution. To complete the proof, it remains to understand that the fundamental value function is continuous on \([0, T] \times J\). This follows from \cite[Lemma~30, p.~116]{freedman}, as we assume \(h \in C_b (\I)\). \qed

	\subsection{Proof of Theorem~\ref{theo: ELMM}}
	
	{\em Step 1:} Take a function \(f \in C_b ( \g (\I))\) that is a dc function with \(\rd  f'_+ = 2k \vd  \tm\) for \(k \in C_b (\g (J))\), i.e., 
	\[
	f'_+ (x) - f'_+ (y) = \int_{(y, x]} 2k (z) \, \tm (\rd z), \quad \forall \, x, y \in \g (\I) \text{ with } y < x. 
	\] 
	By the generalized It\^o formula \cite[Theorem~ VI.1.5]{RY}, \cite[Exercise~VI.1.23]{RY} and the occupation time formula for diffusions \cite[Theorem~V.49.1]{RW2}, we obtain that \(P\)-a.s.
	\begin{align*}
		\rd f (\g (\Y_t)) &= f'_- (\g (\Y_t)) \vd  \g (\Y_t) + \frac{1}{2} \int_{\g (\I)} \rd L^x_t (\g (\Y)) 2k (x) \, \tm (\rd x)
		\\&= f'_- (\g (\Y_t)) \vd  \g (\Y_t) + \int_{J} \rd L^{\g (x)}_t (\g (\Y)) k (\g (x)) \frac{\s'_+ (x)}{\G (x)} \, \m (\rd x)
		\\&= f'_- (\g (\Y_t)) \vd  \g (\Y_t) + \int_{J} \rd L^{\s (x)}_t (\s (\Y)) k (\g (x))  \, \m (\rd x)
		\\&= f'_- (\g (\Y_t)) \vd  \g (\Y_t) + k (\g (\Y_t)) \vd t. \phantom \int
	\end{align*}
	Using \eqref{eq: g (Y) dynamics}, we compute further
	\begin{align*} 
		d f (\g (\Y_t)) &= f'_- (\g (\Y_t)) \g'_- (\Y_t) \vd  \Y_t - f'_- (\g (\Y_t)) r \Y_t \g'_- (\Y_t) \vd t + k (\g (\Y_t)) \vd t
		\\&= f'_- (\g (\Y_t)) \g'_- (\Y_t) \, (\rd  \Y_t - r \Y_t \vd t) + k (\g (\Y_t)) \vd t.
	\end{align*} 
	By integration by parts, we observe that 
	\begin{align*}
		e^{rt} \vd  \tY_t = \vd  \Y_t - r \Y_t \vd t. 
	\end{align*}
	Putting the pieces together, we obtain that 
	\begin{align} \label{eq: pre 1}
		d f (\g (\Y_t)) &= f'_- (\g (\Y_t)) \g'_- (\Y_t) e^{rt} \vd \tY_t + k (\g (\Y_t)) \vd t.
	\end{align} 
	
	{\em Step 2:} Assume now that \(Q\) is an ELMM.
	Then, \eqref{eq: pre 1} holds also under \(Q\) and, because \(\tY\) is a local \(Q\)-martingale, the same is true for 
	\[
	f (\g (\Y_t)) - \int_0^t k (\g (\Y_s)) \vd s, \quad t \in [0, T]. 
	\] 
	In fact, by the boundedness of \(f\) and \(k\), this process is even a true martingale.
	Applying the martingale problem for general diffusions, \cite[Lemmata~B.4, B.5]{CU_AAP_25}, locally, we obtain that 
	\[
	Q \circ S^{-1} = Q_{s_0} \text{ on } \mathcal{W}_{T_\I - }, \quad T_\I := \inf \{t \in [0, T] \colon \X_t \not \in \I \}. 
	\] 
	In particular, \(Q_{s_0} (T_\I > T) = Q \circ \Y^{-1} (T_\I > T) = 1\), because \(Q \sim P\). Thus, \(\J = \I\), as otherwise we get a contradiction to \cite[Theorem~1.1]{bruggeman}, and \eqref{eq: ELMM identity} follows. Furthermore, we conclude from \cite[Corollary~2.21]{CU_AAP_25} that the second part of the necessary condition holds. This proves (b) and one implication of (a).
	
	\smallskip
	{\em Step 3:} To complete the proof of (a), it is left to prove the sufficient conditions for the existence of an ELMM. We prove the claim assuming that the underlying filtered probability space is the canonical one. Given this, the existence of an ELMM on the general setup can be deduced precisely as in the ``alternative proofs'' for Corollary~3.11 in \cite[Section~5.7]{CU_FS_25}. To outline the idea, suppose there exists an ELMM for the canonical market and denote the density process with respect to the real-world measure by \((Z_t)_{t \in [0, T]}\). Then, on an arbitrary setup (meaning not necessarily the canonical one), we may define an ELMM through the density \(Z_T (S)\). More precisely, then \((Z_t (S))_{t \in [0, T]}\) is an a.s. strictly positive martingale, and \((\widetilde{S}_t Z_t (S))_{t \in [0, T]}\) is a local martingale, by \cite[Proposition~III.3.8]{JS}, \cite[Theorem~10.37]{Jacod} and \cite[Lemma~5.13]{CU_FS_25}, entailing that \(Z_T (S)\) serves as density for an ELMM again by \cite[Proposition~III.3.8]{JS}. To complete the proof we need to prove the existence of an ELMM in the canonical setup.
	Let \(Q\) be as in \eqref{eq: ELMM identity}. Then, taking into account that \(P\) and \(Q\) have no accessible boundary points, and recalling Remark~\ref{rem. SA}~(c), \cite[Corollary~2.21]{CU_AAP_25} shows that \(Q\) and \(P\) are equivalent on \(\mathcal{W}_T\). Thus, it suffices to understand that \(\tY\) is a \(Q\)-local martingale. By the equivalence of \(P\) and \(Q\), all equations that hold under \(P\) also hold under~\(Q\), and with \eqref{eq: g (Y) dynamics} we get \(Q\)-a.s.
	\begin{align*}
		\vd \tY_t = e^{-rt} \vd  \Y_t - r e^{-rt} \Y_t \vd t = \frac{e^{- rt}}{\g'_- (\Y_t)} \, (\g'_- (\Y_t) \,  \vd \Y_t - r \Y_t \g'_- (\Y_t) \vd t) = \frac{e^{-rt}}{\g'_- (\Y_t)} \vd  \g (\Y_t). 
	\end{align*}
	Now, as \(\g (\Y)\) is a local \(Q\)-martingale by \cite[Proposition~VII.3.5]{RY}, we conclude that \(\tY\) is a local \(Q\)-martingale, too. This completes the proof of (a). \qed

	\subsection{Proof of Theorem~\ref{thm: identification}}
	
	By Lemma~\ref{lem: simple ineq} and Theorem~\ref{theo: main1}, we only have to prove that 
	\[
	E^{Q_{s_0}} \big[ e^{-rT}h (\X_T) \big] \leq x (h, T). 
	\]
	Let \(H \in \Pi_{\text{adm}}\) be such that \(P\)-a.s. \(V^H_0 = x\) and \(\limsup_{t \nearrow T} V^H_{t} \geq h (\Y_T)\). As \(H\) is self-financing, we get from Eq. (7) on p.~644 in \cite{shir} that \(P\)-a.s. 
	\begin{align} \label{eq: from shir}
		e^{- rt} V^H_t  = x+  \int_0^t H^{(1)}_s \vd \tY_s, \quad t < T.
	\end{align}
	Let \(Q\) be an ELMM. By the equivalence of \(Q\) and \(P\), the equality \eqref{eq: from shir} also holds \(Q\)-a.s. Moreover, as \(\oS\) is a continuous local \(Q\)-martingale, the process \((e^{- rt} V^H_t)_{t < T}\)
	is a local \(Q\)-martingale and, by its boundedness from below (which follows from the admissibility of \(H\)), even a \(Q\)-supermartingale. Consequently, by \cite[Theorems~3.19 and 3.25]{LeGall} the limit \(V^H_{T-} := \lim_{t \nearrow T} V^H_t\) exists \(Q\)-a.s. and 
	\[
	E^Q \big[ e^{- rT} V^H_{T-} \big] \leq E^Q \big[ V^H_0 \big] = x. 
	\] 
	Using that \(P\)-a.s., and by the equivalence \(P \sim Q\) also \(Q\)-a.s., \(\limsup_{t \nearrow T} V^H_{t} \geq h (\Y_T)\), we conclude that 
	\[
	E^Q \big[ e^{- rT} h (\Y_T) \big] \leq x.
	\] 
	As \(Q \circ \Y^{-1} = Q_{s_0}\) by the final part of Theorem~\ref{theo: ELMM}, the theorem is proved. \qed

    \section*{Acknowledgments}
	
	Alexis Anagnostakis was supported by the Centro de Modelamiento Matemático (CMM) BASAL fund FB210005 for center of excellence from ANID-Chile.

    \bibliographystyle{abbrv}  
	\bibliography{bibfile}
	
\end{document}